\documentclass[journal]{IEEEtran}

\usepackage[left=1in, right=1in, top=1in, bottom=1in]{geometry}

\usepackage{amsmath,amssymb,amsfonts,graphicx,amsthm,nicefrac,mathtools}
\usepackage{hyperref}
\usepackage[numbers]{natbib}
\usepackage[english]{babel}
\usepackage{times}
\usepackage[T1]{fontenc}
\usepackage{bbm}
\usepackage{color}
\usepackage{booktabs} 
\usepackage{caption}
\usepackage{subcaption}
\usepackage{xspace}

\newtheorem{theorem}{Theorem}
\newtheorem{lemma}{Lemma}
\newtheorem{proposition}{Proposition}
\newtheorem{corollary}{Corollary}

\theoremstyle{definition}
\newtheorem{definition}{Definition}

\theoremstyle{remark}
\newtheorem{remark}{Remark}

\makeatletter
\newtheorem*{rep@theorem}{\rep@title}
\newcommand{\newreptheorem}[2]{%
\newenvironment{rep#1}[1]{%
 \def\rep@title{#2 \ref{##1}}%
 \begin{rep@theorem}}%
 {\end{rep@theorem}}}
\makeatother

\newreptheorem{theorem}{Theorem}
\newreptheorem{lemma}{Lemma}
\newreptheorem{corollary}{Corollary}
\newreptheorem{proposition}{Proposition}

\newcommand{\subsig}[1]{{#1}_{\mathrm{sig}}} 
\newcommand{\subcor}[1]{{#1}_{\mathrm{cor}}} 

\newcommand{\gcomplexity}[1]{\omega\left({#1}\right)} 
\newcommand{\oldgwidth}[1]{\nu\left({#1}\right)} 
\newcommand{\gdist}[1]{\eta\left({#1}\right)} 
\newcommand{\sqgcomplexity}[1]{\omega^2\left({#1}\right)} 
\newcommand{\sqgdist}[1]{\eta^2\left({#1}\right)} 

\newcommand{\tcone}{T} 
\newcommand{\cltcone}{\overline{\tcone}} 
\newcommand{\stcone}{T_{\mathrm{sig}}} 
\newcommand{\ctcone}{T_{\mathrm{cor}}} 
\newcommand{\jtcone}{T^{\lambda}_{\mathrm{joint}}} 
\newcommand{\ncone}{N} 

\newcommand{\snorm}[1]{\norm{#1}_{\mathrm{sig}}} 
\newcommand{\cnorm}[1]{\norm{#1}_{\mathrm{cor}}} 

\newcommand{\twoball}[1]{\mathbb{B}^{#1}} 

\newcommand{\dist}[2]{\mathrm{dist}\!\left({#1},{#2}\right)} 

\newcommand{\figref}[1]{Figure~\ref{fig:#1}}

\newcommand{\secref}[1]{Section~\ref{sec:#1}}
\newcommand{\appref}[1]{Appendix~\ref{app:#1}}

\newcommand{\lemref}[1]{Lemma~\ref{lem:#1}}

\newcommand{\propref}[1]{Proposition~\ref{prop:#1}}
\newcommand{\propsref}[1]{Propositions~\ref{prop:#1}}
\newcommand{\propssref}[1]{\ref{prop:#1}}
\newcommand{\thmref}[1]{Theorem~\ref{thm:#1}}

\newcommand{\corref}[1]{Corollary~\ref{cor:#1}}
\newcommand{\tabref}[1]{Table~\ref{tab:#1}}
\newcommand{\tabsref}[1]{Tables~\ref{tab:#1}}
\newcommand{\tabssref}[1]{\ref{tab:#1}}

\newcommand{\qtext}[1]{\quad\text{#1}\quad} 

\newcommand{\ignore}[1]{}

\newcommand{\diff}[1]{\text{d}{#1}} 

\newcommand{\eps}{\epsilon}

\newcommand{\inner}[2]{\langle{#1},{#2}\rangle} 
\newcommand{\norm}[1]{\left\lVert{#1}\right\rVert}
\newcommand{\opnorm}[1]{\norm{{#1}}_{\operatorname{op}}} 

\newcommand{\PP}[1]{\mathbb{P}\left\{{#1}\right\}} 
\newcommand{\EE}[1]{\mathbb{E}\left[{#1}\right]} 
\newcommand{\Ep}[2]{\mathbb{E}_{#1}\left[{#2}\right]}

\renewcommand{\O}[1]{\mathcal{O}\left({#1}\right)}

\def\R{\mathbb{R}}
\newcommand\sphere[1]{\mathbb{S}^{{#1}-1}}

\newcommand{\iid}[0]{i.i.d.\xspace}

  \newcommand{\iidsim}{\stackrel{\mathrm{iid}}{\sim}}

\DeclareMathOperator{\sign}{sign}

\DeclareMathOperator{\supp}{support}
\DeclareMathOperator{\cone}{cone} 

\DeclareMathOperator*{\argmin}{arg\,min}

\renewcommand{\Pr}[2]{\mathcal{P}_{{#1}}\left({#2}\right)}
\newcommand{\Prp}[2]{\mathcal{P}_{{#1}}^{\perp}\left({#2}\right)}

\title{Corrupted Sensing: Novel Guarantees for Separating Structured Signals}
\author{Rina~Foygel and Lester~Mackey\\
Department of Statistics, Stanford University}

\begin{document}

\maketitle

\begin{abstract}
We study the problem of corrupted sensing, a generalization of compressed sensing in which one aims to recover a signal from a collection of corrupted or unreliable measurements.
While an arbitrary signal cannot be recovered in the face of arbitrary corruption, tractable recovery is possible when both signal and corruption are suitably structured. 
We quantify the relationship between signal recovery and two geometric measures of structure, the Gaussian complexity of a tangent cone and the Gaussian distance to a subdifferential.
We take a convex programming approach to disentangling signal and corruption, analyzing both penalized programs that trade off between signal and corruption complexity,
and constrained programs that bound the complexity of signal or corruption when prior information is available.
In each case, we provide conditions for exact signal recovery from structured corruption and stable signal recovery from structured corruption with added unstructured noise.
Our simulations demonstrate close agreement between our theoretical recovery bounds and the sharp phase transitions observed in practice. 
In addition, we provide new interpretable bounds for the Gaussian complexity of sparse vectors, block-sparse vectors, and low-rank matrices, which lead to sharper guarantees of recovery when combined with our results and those in the literature.
\end{abstract}

\begin{keywords}
Corrupted sensing, compressed sensing, deconvolution, error correction, structured signal, sparsity, block sparsity, low rank, atomic norms, $\ell_1$ minimization.
\end{keywords}

\section{Introduction}
\label{sec:intro}
In the corrupted sensing problem, our goal is to recover a structured signal from a collection of potentially corrupted measurements.
Recent years have seen a flurry of interest in specific instances of this problem, including sparse vector recovery from sparsely corrupted measurements~\cite{Li13} and the recovery of low-rank matrices from sparse corruption~\cite{chandrasekaran2011rank,CandesLiMaWr09}.
The former arises in applications such as face recognition~\cite{wright2009robust} and in the analysis of sensor network data~\cite{haupt2008compressed};  the latter arises in problems ranging from latent variable modeling~\cite{chandrasekaran2011rank} to video background subtraction~\cite{CandesLiMaWr09}.
In the present work, we are more broadly interested in the deconvolution of an arbitrary signal-corruption pair.
While the problem is generally ill-posed, one might hope that recovery is possible when both signal and corruption are suitably structured.

Corrupted sensing can be viewed as a generalization of the compressed sensing problem, in which one aims to recover a structured signal from a relatively small number of measurements. 
This problem is ubiquitous in modern applications, where one is often interested in estimating a high-dimensional signal from a number of measurements far smaller than the ambient dimension. 
It is now common practice, when the signal has underlying low-dimensional structure,
to promote that structure via a convex penalty and thereby achieve accurate estimation in the face of extreme undersampling.
Two examples extensively studied in the literature are the recovery of sparse vectors via $\ell_1$ norm penalization \cite{CandesTao,Donoho,Cai} and the recovery of low-rank matrices via trace norm penalization \cite{Fazel,FazelHindiBoyd, CandesRe09,CandesTaoMatrix,Recht11}. 

Recent work by \citet{ChandrasekaranRePaWi12} formulates a general framework for this compressed sensing problem, 
in which the complexity of an arbitrary structured signal $x^{\star} \in \R^p$ is encoded in the geometric properties of a norm $\snorm{\cdot}$ used to estimate the signal.
Specifically, given a vector of $n$ noisy measurements
$y = \Phi x^{\star} + z$,
where $\Phi$ is a Gaussian measurement matrix and $z$ is a bounded noise vector, their work gives conditions for when $x^{\star}$ can be recovered from the convex program
\[\min_x\left\{\snorm{x}:\norm{y - \Phi x}_2\leq \delta\right\}\;,\]
for some bound $\delta$ on the noise level $\norm{z}_2$.
In the noiseless setting where $\delta = 0$, the authors' geometric analysis shows that 
\begin{equation}\label{eqn:C-et-al-result}
n = \sqgcomplexity{\stcone\cap\twoball{p}}+1
\end{equation}
measurements suffice to recover $x^{\star}$ \emph{exactly}. 
Here, $\stcone$ is a convex cone in $\R^p$ induced by $x^{\star}$ and $\snorm{\cdot}$, and $\sqgcomplexity{\stcone\cap\twoball{p}}$ is a specific measure of the size of this cone, defined in \secref{convexgeo}. 
In the noisy setting where $\delta > 0$, the same analysis shows that 
$\O{\sqgcomplexity{\stcone\cap\twoball{p}}}$ measurements
suffice to recover $x^{\star}$ \emph{stably}, that is, with error proportional to the noise level.

Our work extends that of \citet{ChandrasekaranRePaWi12} to a more challenging setting, in which signal measurements may not be trustworthy.
Specifically, we allow our measurements
\[y = \Phi x^{\star} + v^{\star} + z\]
to be corrupted by an unknown but structured vector $v^{\star}$,
and bound the sample size $n$ needed to recover $x^{\star}$ and $v^{\star}$ exactly or stably, using convex optimization.
As an example, if $z = 0$ and $v^{\star}$ is $( n\cdot\gamma)$-sparse, so that a fraction $\gamma$ of our linear measurements are arbitrarily corrupted, then our analysis guarantees exact recovery as soon as $n$ exceeds
\[\frac{\sqgcomplexity{\stcone\cap\twoball{p}}}{\nicefrac{2}{\pi}\cdot (1-\gamma)^2}\]
plus an explicit smaller-order term.
This provides a close parallel with \citet{ChandrasekaranRePaWi12}'s result  \eqref{eqn:C-et-al-result} for the corruption-free setting and gives an explicit scaling in terms of the corruption complexity $\gamma$.
More generally, our analysis characterizes recovery from a wide variety of corruption structures, including block-sparse, low-rank, and binary corruption, by appealing to a unified geometric treatment of the complexity of the vector $v^{\star}$.

\subsection{Problem formulation and methodology} \label{sec:problem}
In the formal corrupted sensing problem, we observe a measurement vector $y = \Phi x^{\star} + v^{\star} + z$ comprised of four components:
\begin{itemize}
\item \textbf{The structured signal $x^{\star}\in\R^p$}, our primary target for recovery.  
Our recovery ability will depend on the complexity of $x^{\star}$ with respect to a given norm $\snorm{\cdot}$.\footnote{We focus on norms due to their popularity in structured estimation, but any convex complexity measure would suffice.}
Common examples of structured signals include sparse vectors, which exhibit low complexity with respect to the $\ell_1$ norm,
and low-rank matrices which exhibit low complexity under the trace norm (the sum of the matrix singular values).
Our specific notion of complexity will be made precise in \secref{convexgeo}.

\item \textbf{The structured corruption $v^{\star}\in\R^n$}, our secondary recovery target. Our recovery ability will depend on the complexity of $v^{\star}$ with respect to a second norm, $\cnorm{\cdot}$.

\item \textbf{The unstructured noise $z\in\R^n$}, satisfying $\norm{z}_2\leq \delta$ for known $\delta \geq 0$.
We make no additional assumption about the distribution or structure of $z$.
Our estimation error bounds for $(x^{\star},v^{\star})$ will grow in proportion to $\delta$.

\item \textbf{The measurement matrix $\Phi \in \R^{n\times p}$}, consisting of \iid Gaussian entries
\[\Phi_{ij}\iidsim N(0,\nicefrac{1}{n})\] 
as in \cite{ChandrasekaranRePaWi12}.
Throughout, we treat $x^{\star}$ and $v^{\star}$ as fixed vectors chosen independently of $\Phi$.\footnote{
In some settings, we can allow for $x^{\star}$ and $v^{\star}$ to be selected adversarially after the matrix $\Phi$ is generated with a similar theoretical analysis but do not present this work here. 
In the literature, \citet{ChandrasekaranRePaWi12} also treat the signal $x^{\star}$ as fixed, while \citet{McCoyTr12} give an additional deconvolution guarantee that holds universally over all low-complexity signals via a union bound argument. 
}
However, the unstructured noise $z$ need not be independent of $\Phi$ and in particular may be chosen adversarially after $\Phi$ is generated.
\end{itemize}

Our goal is tractable estimation of $x^{\star}$ and $v^{\star}$ given knowledge of $y$, $\delta$, and $\Phi$.
To this end, we consider two convex programming approaches to disentangling signal and corruption.
The first approach penalizes a combination of signal and corruption complexity, subject to known measurement constraints:
\begin{align} \label{eqn:arg-pen-prob}
\min_{x,v}\left\{\snorm{x}+\lambda\cnorm{v}:\norm{y-(\Phi x + v)}_2\leq \delta\right\}\;.
\end{align}
In \secref{theory}, we will discuss specific settings of the parameter $\lambda$, which trades off between the two penalties.
The second approach makes use of available prior knowledge of either $\snorm{x^{\star}}$ or $\cnorm{v^{\star}}$ 
(for example, a binary vector $x^{\star}$ always satisfies $\norm{x^{\star}}_\infty=1$) to explicitly constrain the signal complexity via
\begin{align} \label{eqn:arg-sig-bound-prob}
\min_{x,v}\left\{\cnorm{v}:\snorm{x}\leq \snorm{x^{\star}},\norm{y-(\Phi x + v)}_2\leq \delta\right\}
\end{align}
or the corruption complexity using
\begin{align} \label{eqn:arg-cor-bound-prob}
\min_{x,v}\!\left\{\!\snorm{x}:\cnorm{v}\leq \cnorm{v^{\star}},\norm{y-(\Phi x + v)}_2\leq \delta\!\right\}.
\end{align}

We note that \citet{McCoyTr12} study a similar, noiseless setting in which $n = p$ and
\[y = Ux^{\star}+v^{\star}\]
for $U$ a uniformly random orthogonal matrix. 
The authors assume that $\cnorm{v^{\star}}$ is known in advance and use spherical integral geometry to characterize the exact recovery of $(x^{\star},v^{\star})$ via the convex program
\begin{align*}
\min_{x,v}\left\{\snorm{x}:\snorm{v}\leq \cnorm{v^{\star}},y=Ux+v\right\}\;.
\end{align*}
Our novel analysis treats both constrained and penalized optimization, provides stable recovery results in the presence of unstructured noise, and covers both the high-dimensional setting ($n < p$) and the overcomplete setting ($n \geq p$).

\subsection{Roadmap}
The remainder of the paper is organized as follows.
In \secref{convexgeo}, we review several concepts from convex geometry that are used throughout the paper and discuss the notions of {\em Gaussian complexity} and {\em Gaussian distance} that form the main ingredients in our recovery results.
\secref{theory} presents our main results for both constrained and penalized convex recovery and gives a brief sketch of our proof strategy.
We apply these results to several specific problems in \secref{applications}, including the problem of \emph{secure and robust channel coding}, and compare our results with those in the literature.
Our experiments with simulated data, summarized in \secref{experiments}, demonstrate a close agreement between our theory and the phase transitions for successful signal recovery observed in practice.
We conclude in \secref{conclusions} with a discussion of our results and several directions for future research. 
All proofs are deferred to the appendices.

\subsection{Notation}
Throughout, we write $\mu_n$ to represent the expected length of an $n$-dimensional vector with independent standard normal entries (equivalently, $\mu_n$ is the mean of the $\chi_n$ distribution)\footnote{The same quantity was represented by $\lambda_n$ in \cite{ChandrasekaranRePaWi12}.}.  It is known that $\mu_1 =\sqrt{\nicefrac{2}{\pi}}$ while $\mu_n\approx\sqrt{n}$ for large $n$; in fact, the expectation is tightly bounded by
\[\sqrt{n-1/2} < \mu_n < \sqrt{n}\]
for all $n$ \citep[Thm.~2]{Chu62}.
In addition, we let $g$ represent a random vector in $\R^p$ with \iid standard Gaussian entries.

\section{Convex Geometry}
\label{sec:convexgeo}

In this section, we review the key concepts from convex geometry that underlie our analysis.
Hereafter, $\twoball{p}$ and $\sphere{p}$ will denote the unit ball and the unit sphere in $\R^p$ under the $\ell_2$ norm.
Throughout the section, we reference a generic norm $\norm{\cdot}$ on $\R^p$, evaluated at a generic point $x\in\R^p\backslash\{0\}$,
but we will illustrate each concept with the example of the $\ell_1$ norm $\norm{\cdot}_1$ on $\R^p$ and an $s$-sparse vector $x_{\text{sparse}}\in\R^p$.
Ultimately, we will apply the results of this section to the norms $\snorm{\cdot}$ on $\R^p$ evaluated at the signal vector $x^{\star}$ and $\cnorm{\cdot}$ on $\R^n$ evaluated at the corruption vector $v^{\star}$.

\subsection{The subdifferential}

The {\em subdifferential} of $\norm{\cdot}$ at $x$ is the set of vectors
\begin{multline*}\partial\norm{x}=\big\{w\in\R^p: \\\norm{x+d}\geq \norm{x} + \inner{w}{d}\text{ for all }d\in\R^p\big\}\;.\end{multline*}
In our running example of the $\ell_1$ norm and the sparse vector $x_{\text{sparse}}$, we have
\begin{multline} \label{eqn:sparse-subdifferential}
\partial\norm{x_{\text{sparse}}}_1=\sign(x_{\text{sparse}}) +\{w \in\R^p: \\\supp(w) \cap \supp(x_{\text{sparse}})=\emptyset, \norm{w}_\infty \leq 1\},\;.
\end{multline}

\subsection{Tangent cones and normal cones}
Our analysis of recovery under corrupted sensing will revolve around two notions from convex geometry: the \emph{tangent cone} and the \emph{normal cone}.
We define the tangent cone\footnote{
We adopt the same terminology for this cone as used by \citet{ChandrasekaranRePaWi12}; \citet{McCoyTr12} call it the ``feasible cone'' since it is the cone of feasible directions under the constraint $\norm{x+w}\leq\norm{x}$. This definition of ``tangent cone'' nearly coincides with the use of the term in convex geometry---considering the convex subset $\{\norm{z}\leq\norm{x}\}$ (that is, the unit ball of the norm $\norm{\cdot}$, rescaled by $\norm{x}$), the tangent cone as defined in convex geometry is given by $\cltcone$, the closure of $T$.
} to $\norm{\cdot}$ at $x$ as the set of descent (or non-ascent) directions of $\norm{\cdot}$ at $x$:
\[
\tcone=\left\{w:\norm{x+c\cdot w}\leq \norm{x}\text{ for some }c>0\right\}\;.
\]
By convexity of the norm $\norm{\cdot}$, $\tcone$ is a convex cone. In our running example, the tangent cone is given by
\begin{multline}\label{eqn:ConeExample}
\tcone_{\text{sparse}}=\big\{w:\norm{w_{\supp(x_{\text{sparse}})^c}}_1\leq\\ -\inner{w}{\mathrm{sign}(x_{\text{sparse}})}\big\}\;.
\end{multline}

The normal cone to $\norm{\cdot}$ at $x$ is the {polar} of the tangent cone, given by
\[
\ncone 
	= \left\{w:\inner{w}{u} \leq 0 \text{ for all }u\in \tcone\right\},
\]
and may equivalently be written as the conic hull of the subdifferential $\partial\norm{x}$~\citep{penot2000elements}:
\[
\ncone 
	= \cone\{\partial\norm{x}\} = \{ w : w \in t\cdot\partial\norm{x} \text{ for some }t \geq 0\}.
\]
\figref{Cones} illustrates these entities in $\R^2$ for the $\ell_1$ and $\ell_2$ norms.

\begin{figure*}
\centering
\includegraphics[width=5in]{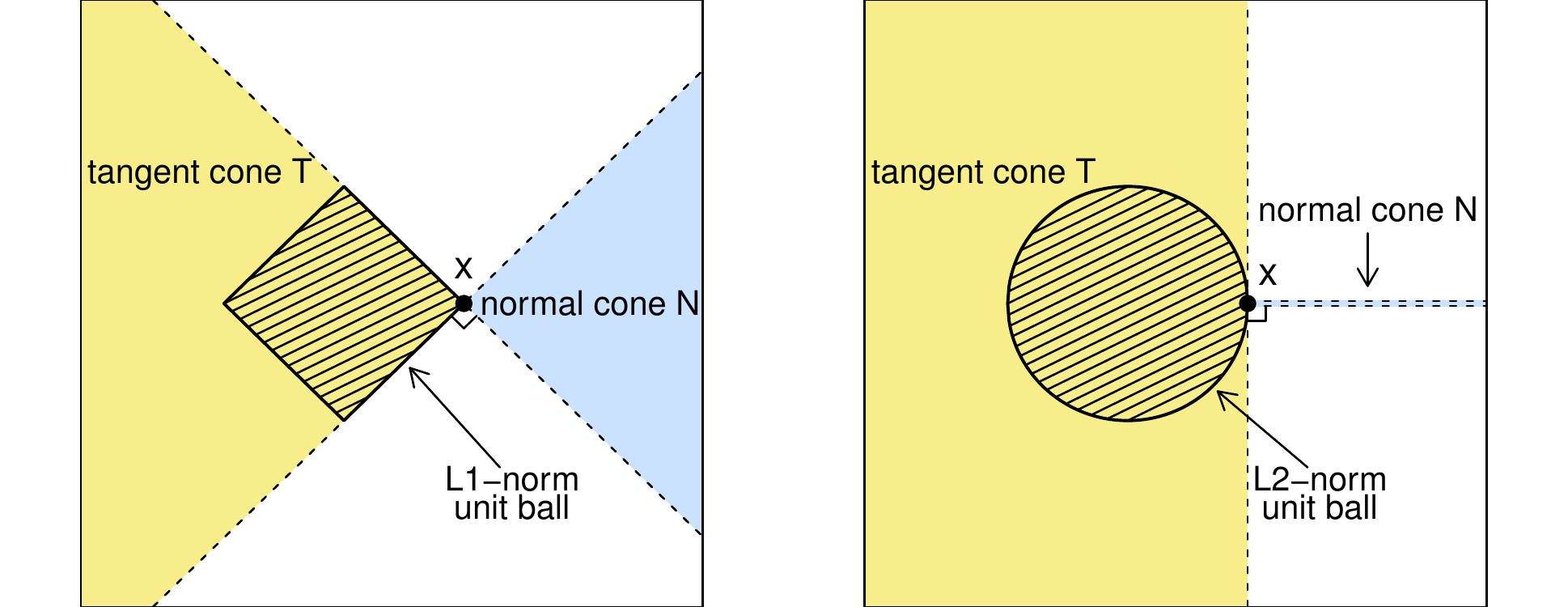}
\caption{Illustrations of the tangent cone and normal cone for the $\ell_1$ norm (left) and the $\ell_2$ norm (right) at the point $x=(1,0)$. (For visual clarity, cones are shifted to originate at the point $x$.) }
\label{fig:Cones}
\end{figure*}

\subsection{Gaussian complexity and Gaussian distance}\label{sec:GaussianWidth}
To quantify the complexity of a structured vector, we adopt two geometric measures of size, the \emph{Gaussian complexity} and the \emph{Gaussian distance}:
\begin{definition}
The \emph{Gaussian squared complexity} $\sqgcomplexity{C}$ of a set $C\subset\R^p$ is given by
\[
\sqgcomplexity{C} = \Ep{g\sim N(0,I_p)}{\left(\sup_{w \in C}\ \inner{g}{w}\right)_+^2}.
\] 
We call the square root of this quantity, $\gcomplexity{C}=\sqrt{\sqgcomplexity{C}}$, the \emph{Gaussian complexity}.
\end{definition}
\begin{definition}
The \emph{Gaussian squared distance} $\sqgdist{C}$ of a set $C\subset\R^p$ is given by
\[
\sqgdist{C} = \Ep{g\sim N(0,I_p)}{\inf_{w \in C}\norm{g - w}^2_2}. 
\] 
We call the square root of this quantity, $\gdist{C}=\sqrt{\sqgdist{C}}$, the \emph{Gaussian distance}.
\end{definition}
\noindent \citet{ChandrasekaranRePaWi12} showed that the Gaussian squared complexity of a restricted tangent cone $\sqgcomplexity{\tcone\cap\twoball{p}}$ determines a sufficient sample size for signal recovery from uncorrupted measurements.%
\footnote{\citet{ChandrasekaranRePaWi12} express their results in terms of the complexity measure
$\oldgwidth{\tcone\cap\sphere{p}}\coloneqq \EE{\sup_{w \in \tcone\cap\sphere{p}}\ \inner{g}{w}}$ (known as the Gaussian width),
which is very slightly smaller than $\gcomplexity{\tcone\cap\twoball{p}}$. 
Nevertheless, the upper bounds developed in \cite{ChandrasekaranRePaWi12} for $\oldgwidth{\tcone\cap\sphere{p}}$ hold also for $\gcomplexity{\tcone\cap\twoball{p}}$, a quantity which arises more naturally from our theory.}
We will establish analogous results for our corrupted sensing setting in \secref{theory}.
 
To obtain interpretable sample size bounds in terms of familiar parameters, it is often necessary to bound $\gcomplexity{\tcone\cap\twoball{p}}$.
\citet{ChandrasekaranRePaWi12} describe a variety of methods for obtaining such bounds.
One especially profitable technique, which can be traced back to \cite{Stojnic09,stojnic2009reconstruction}, relates this Gaussian complexity to the Gaussian distance of the scaled subdifferential of $\norm{\cdot}$ at $x$, by way of the normal cone $\ncone$:
\begin{multline*} \label{eqn:TangentWidthDistance1}
\sqgcomplexity{\tcone\cap\twoball{p}}=\sqgdist{\ncone}=\EE{\inf_{w\in\ncone}\norm{g-w}^2_2}\\=\EE{\min_{t\geq 0}\inf_{w\in t\cdot \partial\norm{x}}\norm{g-w}^2_2}\leq \min_{t\geq 0} \sqgdist{t\cdot\partial\norm{x}}\;.
\end{multline*}
Indeed, many of the complexity bounds in \cite{ChandrasekaranRePaWi12} are derived by bounding $\gdist{t\cdot\partial\norm{x}}$ at a fixed value of $t$. 

In \appref{dist-vs-complexity-proof}, we show that the best choice of $t$ typically yields a bound within a small additive constant of the Gaussian complexity:
\begin{proposition}\label{prop:dist-vs-complexity}
Suppose that, for $x \neq 0$,  $\partial\norm{x}$ satisfies a \emph{weak decomposability} assumption:
\begin{equation}\label{eqn:subdiff-decomp1}
\exists w_0\in\partial\norm{x} \text{ s.t. }\inner{w-w_0}{w_0}=0, \ \forall w\in\partial\norm{x}\;.\end{equation}
Then
\[\min_{t\geq0}\gdist{t\cdot\partial\norm{x}}\leq\gcomplexity{\tcone\cap\twoball{p}}+6\;.\]
\end{proposition}
\noindent
\begin{remark}
The weak decomposability assumption \eqref{eqn:subdiff-decomp1} is satisfied in all the examples considered in this paper and is weaker than the decomposability assumptions often found in the literature (see, e.g., \cite{negahban2012unified}).
\end{remark}
\noindent
A complementary result relating Gaussian distance and Gaussian complexity appears in \citet[Thm.~4.5]{AmelunxenLoMcTr13}:
\begin{proposition}[{\citet[Thm.~4.5]{AmelunxenLoMcTr13}}]
\label{prop:dist-vs-complexity2}
For any $x\neq 0$, 
\begin{align*}
\min_{t\geq0}\sqgdist{t\cdot\partial\norm{x}} 
\leq
\sqgcomplexity{\tcone\cap\twoball{p}} + \frac{2\sup_{s\in \partial\norm{x}}\norm{s}}
{\norm{x}/\norm{x}_2}.
\end{align*}
\end{proposition}
\noindent
Together, \propsref{dist-vs-complexity} and \propssref{dist-vs-complexity2} demonstrate that the optimized Gaussian distance $\min_{t\geq0}\gdist{t\cdot\partial\norm{x}}$ offers a faithful approximation to the Gaussian complexity $\gcomplexity{\tcone\cap\twoball{p}}$.

In the case of our $\ell_1$ norm example, \citet{ChandrasekaranRePaWi12} choose $t = \sqrt{2\log(p/s)}$ and show that\footnote{\citet{ChandrasekaranRePaWi12} report a slightly smaller bound based on a minor miscalculation.}
\begin{multline}\label{eqn:SparseConeBound_old}
\sqgcomplexity{\tcone_{\text{sparse}}\cap \twoball{p}}\leq \sqgdist{t\cdot\partial\norm{x_{\text{sparse}}}_1} \\\leq 2s\log(p/s)+\frac{3}{2}s\;.
\end{multline}
We will show that the alternative scaling $t' = \sqrt{2/\pi}(1-s/p)$ yields the bound
\begin{multline}\label{eqn:SparseConeBound_new}
\sqgcomplexity{\tcone_{\text{sparse}}\cap \twoball{p}}
	\leq \sqgdist{t'\cdot\partial\norm{x_{\text{sparse}}}_1}\\
	\leq p\cdot \left(1-\nicefrac{2}{\pi}\cdot(1-\nicefrac{s}{p})^2\right)\;,
\end{multline}
which is a tighter than \eqref{eqn:SparseConeBound_old} whenever $s\geq 0.07\cdot p$. 
Moreover, this new bound, unlike \eqref{eqn:SparseConeBound_old}, is strictly less than $p$ for any $s<p$; this has important implications for the recovery of structured signals from nearly dense corruption.

Both bounds \eqref{eqn:SparseConeBound_old} and \eqref{eqn:SparseConeBound_new} are ultimately derived from the expected squared distance definition
\begin{align*}
\sqgdist{t\cdot\partial\norm{x}} &= \EE{\inf_{w\in t\cdot\partial\norm{x}}  \norm{g - w}_2^2}.
\end{align*}
In the $\ell_1$ setting, this expected squared distance has a manageable representation (see \cite[App.~C]{ChandrasekaranRePaWi12}) that can be directly minimized over $t$:
\begin{multline}\label{eqn:SparseDistSquared}
\min_{t\geq0}\EE{\inf_{w\in t\cdot\partial\norm{x_{\text{sparse}}}_1} \norm{g - w}_2^2} 
= \min_{t\geq0} s(1+t^2) \\ {}+\frac{2(p-s)}{\sqrt{2\pi}}\left((1+t^2)\int_t^\infty e^{-c^2/2}\ \diff{c} -te^{-t^2/2}\right).
\end{multline}
The resulting bound on the Gaussian complexity is tighter
 than either \eqref{eqn:SparseConeBound_old} or \eqref{eqn:SparseConeBound_new}, albeit without an interpretable closed form.
We compare the distance bounds arising from these three calculations in \figref{SparseConeWidth}.
Notice that the new closed-form bound \eqref{eqn:SparseConeBound_new} closely approximates the optimal squared distance bound \eqref{eqn:SparseDistSquared} for a wide range of sparsity levels,
while the prior closed-form bound \eqref{eqn:SparseConeBound_old} only dominates \eqref{eqn:SparseConeBound_new} in the extreme sparsity regime.
\begin{figure*}[Sparse recovery distance estimates]
\centering
\includegraphics[width=4in]{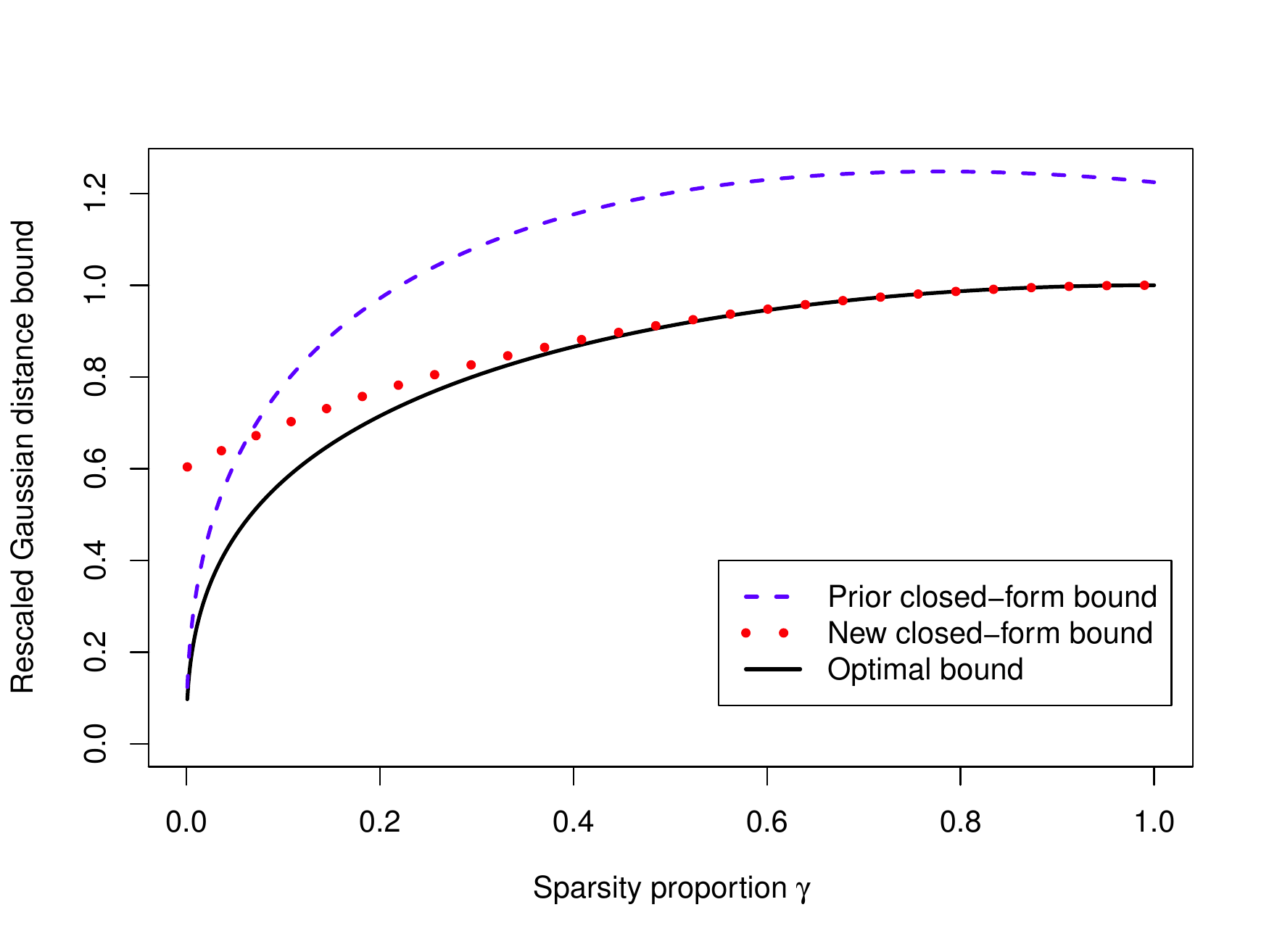}
\caption{Three upper bounds on the rescaled Gaussian distance $\min_{t\geq 0}\gdist{t\cdot\partial\norm{x_\text{sparse}}_1}/\sqrt{p}$ for an $s$-sparse vector $x_\text{sparse}\in\R^p$ as a function of the sparsity proportion $\gamma=s/p$.
The prior bound comes from \eqref{eqn:SparseConeBound_old}, the new bound from \eqref{eqn:SparseConeBound_new}, and the optimal bound comes from the exact expression for expected squared distance given in \eqref{eqn:SparseDistSquared}.
}
\label{fig:SparseConeWidth}
\end{figure*}

It is instructive to contrast the complexity of $\tcone_{\text{sparse}}\cap \twoball{p}$ in our example with that induced by the $\ell_2$ norm on $\R^p$. 
The $\ell_2$ norm tangent cone at any point $x$ will be a half-space (as illustrated in \figref{Cones}), and its Gaussian complexity will be approximately $\sqrt{p}$, irrespective of the structure of $x$.
Meanwhile, any sparse vector will exhibit a much smaller Gaussian complexity under the $\ell_1$ norm tangent cone, as illustrated in \eqref{eqn:SparseConeBound_old} and \eqref{eqn:SparseConeBound_new}. 
We will see that the choice of norm and the attendant reduction in geometric complexity are critical for achieving signal recovery in the presence of corruption.

\subsection{Structured vectors and structure-inducing norms}
\label{sec:examples}
While our analysis is applicable to any signal coupled with any norm, it is most profitable when the signal exhibits low complexity under its associated norm.
In this section, we will highlight a number of important structured signal classes and norms under which these signals have low complexity (additional examples can be found in \cite{ChandrasekaranRePaWi12}):
 \begin{itemize}
 \item\textbf{Sparse vectors:} As discussed throughout the section, an appropriate norm for recovering sparse vectors in $\R^p$ is the $\ell_1$ norm, $\norm{\cdot}_1$.
 \item\textbf{Low-rank matrices:} We will see that low-rank matrices in $\R^{m_1\times m_2}$ exhibit low geometric complexity under the trace norm $\norm{\cdot }_*$, given by the sum of the singular values of a matrix.
 \item\textbf{Binary vectors:} A binary vector of length $p$ has relatively low complexity under the $\ell_{\infty}$ norm, $\norm{\cdot}_{\infty}$, which has a known value of $1$ for any such vector.
 \item\textbf{Block-sparse vectors:} We will show that a block-sparse vector in $\R^p$, supported on a partition of $\{1,\dots,p\}$ into $m$ blocks of size $p/m$, has low complexity under the $\ell_1/\ell_2$ norm
 \[
\norm{x}_{\ell_1/\ell_2} = \sum_{b=1}^m\norm{x_{V_b}}_2,
\]
where $V_b\subset\{1,\dots,p\}$ is the $b$th block.
 \end{itemize}

\tabref{WidthBounds} gives bounds on the squared complexity of the tangent cone for each of these structured settings (recall that $\mu_k$ is the mean of a $\chi_k$ distribution).
These estimates will be useful for studying signal recovery via the constrained convex optimization problems \eqref{eqn:arg-sig-bound-prob} and \eqref{eqn:arg-cor-bound-prob}.
\tabref{DistBounds} highlights those results obtained by bounding the larger quantity $\sqgdist{t\cdot\partial\norm{x}}$
for a fixed scaling $t$ of the subdifferential.
We will make use of these Gaussian squared distance estimates and their associated settings of $t$ when studying penalized convex recovery via \eqref{eqn:arg-pen-prob}.
The new results in \tabsref{WidthBounds} and \tabssref{DistBounds} are derived in \appref{norms}.

\renewcommand{\arraystretch}{1.5}
\begin{table*}
\centering
\caption{Upper bounds on the restricted tangent cone Gaussian squared complexity $\sqgcomplexity{\tcone\cap\twoball{p}}$ 
for various structured vectors and structure-inducing norms. Prior bounds on Gaussian complexity are due to~\cite{ChandrasekaranRePaWi12}. 
}
\begin{tabular}{llcc}
\toprule
Structure & $\norm{\cdot}$  & Prior Gaussian sq.\ complexity bound  & New Gaussian sq.\ complexity bound \\
\hline\hline
$s$-sparse $p$-dim.\ vector   & $\norm{\cdot}_1$ & ${2s\log(\frac{p}{s}) + \frac{3}{2}s}$ & ${p}\left({1 - \frac{2}{\pi}(1-\frac{s}{p})^2}\right)$ \\\hline
\begin{tabular}{@{}l@{}}$s$-block-sparse vector\\($m$ blocks of size $k=\nicefrac{p}{m}$)\end{tabular} & $\norm{\cdot}_{\ell_1/\ell_2}$ & -  & 
\begin{tabular}{@{}c@{}} $\left.{4s\log(\frac{m}{s}) + (\frac{1}{2}+3k)s},\right.$\\  $\left.{p}\left({1 - \frac{\mu_k^2}{k}(1-\frac{s}{m})^2}\right)\right.$\end{tabular}
\\\hline
\begin{tabular}{@{}l@{}} $r$-rank $m_1\times m_2$ matrix\\ ($m_1\geq m_2$) \end{tabular}     &  $\norm{\cdot}_*$  &  ${3r(m_1+m_2-r)}$ & ${m_1m_2}\left({1-(\frac{4}{27})^2\left(1-\frac{r}{m_1}\right)\left(1-\frac{r}{m_2}\right)^2}\right)$ \\\hline
binary $p$-dim.\ vector		& $\norm{\cdot}_\infty$		& ${\frac{p}{2}}$ & - \\\hline
\end{tabular}
\label{tab:WidthBounds}
\end{table*}

\renewcommand{\arraystretch}{1.5}
\begin{table*}
\centering
\caption{Upper bounds 
on the scaled subdifferential Gaussian squared distance $\sqgdist{t\cdot\partial\norm{x}}$
for various structured vectors, structure-inducing norms, and settings of $t$.
}
\begin{tabular}{llcc}
\toprule
Structure & $\norm{\cdot}$  & Gaussian squared distance bound & Setting of $t$ used to achieve bound \\
\hline\hline
$s$-sparse $p$-dim.\ vector   & $\norm{\cdot}_1$ &\begin{tabular}{@{}c@{}} ${2s\log(\frac{p}{s}) + \frac{3}{2}s},$\\ ${p}\left({1 - \frac{2}{\pi}(1-\frac{s}{p})^2}\right)$\end{tabular}&\begin{tabular}{@{}c@{}}$\sqrt{2\log(\nicefrac{p}{s})},$\\$\sqrt{\frac{2}{\pi}}(1-\nicefrac{s}{p})$\end{tabular}\\\hline
\begin{tabular}{@{}l@{}}$s$-block-sparse vector\\($m$ blocks of size $k=\nicefrac{p}{m}$)\end{tabular} & $\norm{\cdot}_{\ell_1/\ell_2}$ & 
\begin{tabular}{@{}c@{}} $\left.{4s\log(\frac{m}{s}) + (\frac{1}{2}+3k)s},\right.$\\  $\left.{p}\left({1 - \frac{\mu_k^2}{k}(1-\frac{s}{m})^2}\right)\right.$\end{tabular}
&\begin{tabular}{@{}c@{}}$\sqrt{2\log(\nicefrac{m}{s})}+\sqrt{k},$\\$\frac{\mu_k}{\sqrt{k}}(1-\nicefrac{s}{m})$\end{tabular}
\\\hline
\begin{tabular}{@{}l@{}} $r$-rank $m_1\times m_2$ matrix\\ ($m_1\geq m_2$) \end{tabular}     &  $\norm{\cdot}_*$  & ${m_1m_2}\left({1-(\frac{4}{27})^2\left(1-\frac{r}{m_1}\right)\left(1-\frac{r}{m_2}\right)^2}\right)$ &
$\frac{4}{27}(m_2-r)\sqrt{m_1-r} / m_2$ \\\hline
\end{tabular}
\label{tab:DistBounds}
\end{table*}

\section{Recovery from Corrupted Gaussian Measurements}
\label{sec:theory}

 We now present our theoretical results on the recovery of structured signals from corrupted measurements via convex programming. Both main theorems are proved in \appref{main-proofs}. 
Below, $\stcone$ denotes the tangent cone for $\snorm{\cdot}$ at the true structured signal vector $x^{\star}$, 
and $\ctcone$ designates the tangent cone for $\cnorm{\cdot}$ at the true structured corruption vector $v^{\star}$.
\subsection{Recovery via constrained optimization}
We begin by analyzing the constrained convex recovery procedures 
\begin{align} \label{eqn:cor-bound-prob}
&\min\left\{\snorm{x}:\cnorm{v}\leq\cnorm{v^{\star}},\norm{y-(\Phi x+ v)}_2\leq \delta\right\} \intertext{and} \label{eqn:sig-bound-prob}
&\min\!\left\{\!\cnorm{v}:\snorm{x}\leq\snorm{x^{\star}},\norm{y-(\Phi x+ v)}_2\leq \delta\!\right\}\!,
\end{align}
which are natural candidates for estimating $(x^{\star},v^{\star})$ whenever prior knowledge of $\snorm{x^{\star}}$ or $\cnorm{v^{\star}}$ is available.
Our first result shows that, with high probability,
approximately
\[
\sqgcomplexity{\stcone\cap\twoball{p}}+\sqgcomplexity{\ctcone\cap\twoball{n}}
\] 
corrupted measurements suffice to recover $(x^{\star},v^{\star})$ exactly in the absence of noise $(\delta = 0)$ and stably in the presence of noise $(\delta \neq 0)$, via either of the procedures \eqref{eqn:cor-bound-prob} or \eqref{eqn:sig-bound-prob}. 

\begin{theorem}[General constrained recovery] \label{thm:gen-con-recovery}
If $(\widehat{x},\widehat{v})$ solves either of the constrained optimization problems \eqref{eqn:cor-bound-prob} or \eqref{eqn:sig-bound-prob}, then
\[\sqrt{\norm{\widehat{x}-x^{\star}}^2_2+\norm{\widehat{v}-v^{\star}}^2_2}\leq \frac{2 \delta}{\eps}\]
with probability at least $1-e^{-(\mu_n-\tau-\eps\sqrt{n})^2/2}$,
as long as $\mu_n-\eps\sqrt{n}$ exceeds the success threshold
\[
\tau = \sqrt{\sqgcomplexity{\stcone\cap\twoball{p}}+\sqgcomplexity{\ctcone\cap\twoball{n}}}+\frac{1}{\sqrt{2}} +\frac{1}{\sqrt{2\pi}}\;.
\]
\end{theorem}
\begin{remark}
	In the noiseless setting where $\delta = 0$, \thmref{gen-con-recovery} entails exact recovery of $(x^{\star},v^{\star})$ with probability at least $1-e^{-(\mu_n-\tau)^2/2}$ as long as $\mu_n\geq\tau$.
	When $n=p$ and $\delta = 0$, this conclusion closely resembles the recent results of \cite{McCoyTr12,AmelunxenLoMcTr13}.
	See \secref{related-work} for a more detailed discussion of these related works.
\end{remark}

\subsection{Recovery via penalized optimization}
When prior knowledge of $\snorm{x^{\star}}$ or $\cnorm{v^{\star}}$ is unavailable, one can instead rely on the penalized convex recovery procedure
\begin{align} \label{eqn:pen-prob}
\min\left\{\snorm{x} + \lambda\cnorm{v}:\norm{y-(\Phi x+ v)}_2\leq \delta\right\}
\end{align}
with penalty parameter $\lambda$.
Our next result provides sufficient conditions for exact and stable penalized recovery
and demonstrates how to set the penalty parameter $\lambda$ in practice.

\begin{theorem}[General penalized recovery] \label{thm:gen-pen-recovery}
Fix any $\subsig{t},\subcor{t}\geq 0$. If $(\widehat{x},\widehat{v})$ solves the penalized optimization problem \eqref{eqn:pen-prob} with penalty parameter $\lambda={\subcor{t}}{/\subsig{t}}$, then
\[\sqrt{\norm{\widehat{x}-x^{\star}}^2_2+\norm{\widehat{v}-v^{\star}}^2_2}\leq \frac{2\delta}{\eps}\]
with probability at least $1-e^{-(\mu_n-\tau-\eps\sqrt{n})^2/2}$,
as long as $\mu_n-\eps\sqrt{n}$ exceeds the success threshold
\begin{multline*}
\tau = 2\gdist{\subsig{t}\cdot\partial\snorm{x^{\star}}}+\gdist{\subcor{t}\cdot\partial\cnorm{v^{\star}}} \\+ 3\sqrt{2\pi}+ \frac{1}{\sqrt{2}}+\frac{1}{\sqrt{2\pi}}\;.
\end{multline*}
\end{theorem}
\begin{remark}
	In the noiseless setting ($\delta = 0$), \thmref{gen-pen-recovery} entails exact recovery of $(x^{\star},v^{\star})$ with probability at least $1-e^{-(\mu_n-\tau)^2/2}$ as long as $\mu_n\geq\tau$.
\end{remark}
\begin{remark}
In analogy to \thmref{gen-con-recovery}, one might expect that 
\begin{align}
\label{eqn:pen-hyp-bound}
\mu_n^2\geq{\sqgdist{\subsig{t}\cdot\partial\snorm{x^{\star}}}+\sqgdist{\subcor{t}\cdot\partial\cnorm{v^{\star}}}}
\end{align}
would suffice for high probability penalized recovery.
Indeed, our recovery result would have this form under a more symmetric observation model in which the corruption vector is also multiplied by a random Gaussian matrix (i.e.,\ $y=\Phi x^\star + \Psi v^\star+z$, where $\Phi$ and $\Psi$ are independent).
Our experimental results (see \figref{pen-sparse-sparse}) suggest that while \eqref{eqn:pen-hyp-bound} is nearly the right threshold for recovery, a visible asymmetry exists between signal and corruption, stemming from the fact that only the signal vector is modulated by a Gaussian matrix under our observation model.
\end{remark}

To make use of \thmref{gen-pen-recovery}, it suffices to bound the Gaussian distance terms 
\[\gdist{\subsig{t}\cdot\partial\snorm{x^{\star}}}\text{ and }\gdist{\subcor{t}\cdot\partial\cnorm{v^{\star}}}\]
for suitably chosen subdifferential scalings $\subsig{t}$ and $\subcor{t}$.
Recall, from \secref{GaussianWidth}, that these Gaussian distances represent practical upper bounds on the tangent cone Gaussian complexities, $\gcomplexity{\stcone\cap\twoball{p}}$ and $\gcomplexity{\ctcone\cap\twoball{n}}$; in fact, \propref{dist-vs-complexity} showed that a well-chosen distance bound closely approximates the tangent cone complexity for most norms of interest.
In practice we can choose the distance scalings $\subsig{t}$ and $\subcor{t}$ based on known upper bounds on the Gaussian distance, such as those found in \tabref{DistBounds}.

In the application section, we will see that different choices of $\subsig{t},\subcor{t},$ and $\lambda$ in \thmref{gen-pen-recovery} allow us to recover or improve upon existing recovery results for specialized signal and corruption structure.
In the experiments section, we will choose $\subsig{t},\subcor{t},$ and $\lambda$ based on practical considerations and demonstrate recovery performance that nearly matches recovery rates obtained with prior knowledge of $\snorm{x^{\star}}$ or $\cnorm{v^{\star}}$.

\subsection{Proof overview}
\label{sec:proofs-sketch}
In each of the optimization problems above, the recovery target $(x^{\star},v^{\star})$ and estimate $(\widehat{x},\widehat{v})$ both lie in the feasible set $\{(x,v):\norm{y-(\Phi x + v)}_2\leq \delta\}$, and hence the error $(\widehat{x}-x^{\star},\widehat{v}-v^{\star})$ satisfies $\norm{\Phi(\widehat{x}-x^{\star})+(\widehat{v}-v^{\star})}_2\leq 2\delta$. 
In \appref{main-proofs}, we bound the size of this error with high probability by lower bounding
\begin{equation}\label{eqn:ThingToLowerBound}\min\frac{\norm{\Phi a + b}_2}{\sqrt{\norm{a}^2_2+\norm{b}^2_2}}\;,\end{equation}
where the minimum is taken uniformly over the class of nonzero perturbation vectors $(a,b)\in\R^p\times\R^n$ satisfying either the stronger requirement
\[\snorm{x^{\star}+a}\leq\snorm{x^{\star}}\quad\text{and}\quad\cnorm{v^{\star}+b}\leq\cnorm{v^{\star}}\]
for either of the two constrained recovery procedures, or the weaker requirement
\[\snorm{x^{\star}+a}+\lambda\cnorm{v^{\star}+b}\leq \snorm{x^{\star}}+\lambda\cnorm{v^{\star}}\]
for the penalized recovery procedure. The lower bound on \eqref{eqn:ThingToLowerBound} can then be applied to the error vector $(a,b)=(\widehat{x}-x^{\star},\widehat{v}-v^{\star})$ to bound its $\ell_2$ norm.

\section{Applications and Related Work}
\label{sec:applications}
In this section, we apply the general theory of \secref{theory} to specific signal and corruption structures of interest, drawing comparisons to existing literature where relevant. All results in this section are proved in \appref{proofs-applications}.
\subsection{Binary recovery from sparse corruption and unstructured noise}
\label{sec:con-binary-sparse}
To exemplify the constrained recovery setting, we analyze a communications protocol for secure and robust channel coding proposed by \citet{Wyner79a,Wyner79b}
and studied by \citet{McCoyTr12}.
The aim is to transmit a binary signal securely across a communications channel while tolerating sparse communication errors.
The original protocol applied a random rotation to the binary signal and thereby constrained the length of the transmitted message to equal the length of the input message. 
Here, we take $n$ random Gaussian measurements $\Phi$ for $n$ not necessarily equal to $p$.
This offers the flexibility of transmitting a shorter, cheaper message with reduced tolerance to corruption, or a longer message with increased corruption tolerance.
In addition, unlike past work, we allow our measurements to be perturbed by dense, unstructured noise.

To recover the signal $x^{\star}$ after observing the sparsely corrupted and densely perturbed message $y = \Phi x^{\star} + v^{\star} + z$ with $\norm{z}_2 \leq \delta$, we solve the constrained optimization problem
\begin{align} \label{eqn:con-binary-sparse-prob}
\min_{x,v}\left\{\norm{v}_1:\norm{x}_\infty\leq 1, \norm{y - (\Phi x + v)}_2 \leq \delta \right\},
\end{align}
where we take advantage of the prior knowledge that $\norm{x}_\infty = 1$ for any binary signal.
The following result, which follows from \thmref{gen-con-recovery}, provides sufficient conditions for \emph{exactly} reconstructing $x^{\star}$ using the convex program~\eqref{eqn:con-binary-sparse-prob}. 

\begin{corollary}[Binary signal, sparse corruption]\label{cor:BinaryPlusSparse}
Suppose that $x^{\star}\in \{+1,-1\}^p$ and that $v^{\star}$ is supported on at most $ n\cdot\gamma$ entries. 
Suppose the number of measurements satisfies
\begin{multline*}
\mu_n - 2\delta >\\ \sqrt{\frac{p}{2}+ \sqgcomplexity{T_{\mathrm{sparse}(n, \gamma)}\cap\twoball{n}}} +\frac{1}{\sqrt{2}} +\frac{1}{\sqrt{2\pi}}\; \eqqcolon \tau,
\end{multline*}
where $T_{\mathrm{sparse}(n, \gamma)}$ is the tangent cone of a $( n\cdot\gamma)$-sparse $n$-dimensional vector under the $\ell_1$ norm, known to satisfy
\begin{multline}\label{eqn:SparseConeBound}\sqgcomplexity{T_{\mathrm{sparse}(n, \gamma)}\cap\twoball{n}}\leq {n}\cdot \\\min\left\{{2 \gamma\log\left(\gamma^{-1}\right) +\frac{3}{2} \gamma},{ 1 - \frac{2}{\pi}\left(1-\gamma\right)^2}\right\}\;.\end{multline}
 Then with probability at least $1-e^{-(\mu_n-2\delta-\tau)^2/2}$, the binary vector $x^{\star}$ is exactly equal to $\sign(\widehat{x})$, where $(\widehat{x},\widehat{v})$ is any solution to
the constrained recovery procedure \eqref{eqn:con-binary-sparse-prob}.
\end{corollary}
\begin{remark}
In addition to the interpretable bound \eqref{eqn:SparseConeBound}, $\gcomplexity{T_{\mathrm{sparse}(n, \gamma)}\cap\twoball{n}}$ admits the sharp, computable bound described in \eqref{eqn:SparseDistSquared}.
\end{remark}

It is enlightening to compare \corref{BinaryPlusSparse} with the binary recovery from sparse corruption results of \citet{McCoyTr12}. 
Assuming that $n=p$, that $\delta=0$, and that the measurement matrix is a uniformly random rotation drawn independently of $(x^{\star},v^{\star})$, \citet{McCoyTr12} proved that signal recovery is possible with up to $19.3\%$ of measurements corrupted. In their work, the threshold value of $\gamma_0 = 19.3\%$ was obtained through a numerical estimate of a geometric parameter analogous to the Gaussian complexity $\gcomplexity{T_{\mathrm{sparse}(n, \gamma)}\cap\twoball{n}}$.
For large $n=p$, the additive constants in \corref{BinaryPlusSparse} become negligible, and our result implies exact recovery with some probability when 
\[
\sqgcomplexity{T_{\mathrm{sparse}(n, \gamma)}\cap\twoball{n}}<\left({(1-2\delta)^2-1/2}\right)\cdot{n}.
\]
In the noiseless setting ($\delta = 0$), this precondition for exact recovery becomes $\sqgcomplexity{T_{\mathrm{sparse}(n, \gamma)}\cap\twoball{n}} <{n/2}$.
Our plot of $\ell_1$ complexity bounds (\figref{SparseConeWidth}) shows that this indeed holds when $\gamma < 19.3\%$. 
Thus, our result parallels that of \citet{McCoyTr12} in the noiseless setting.
Unlike the work of \citet{McCoyTr12}, \corref{BinaryPlusSparse} also provides for exact recovery from dense, unstructured noise, and additionally allows for $n\neq p$.

\subsection{General penalized recovery from block-sparse corruption and unstructured noise}\label{sec:gen-pen-examples}

As a first application of our penalized recovery result, \thmref{gen-pen-recovery}, we consider a setting in which the corruption is entry-wise or block-wise sparse, and the signal exhibits an arbitrary structure. While past work has analyzed corruption-free block-sparse signal recovery in the setting of compressed sensing~\citep{stojnic2009reconstruction,EldarKuBo10,ParvareshHa08,CandesRe11},
to our knowledge, \corref{pen-gen-block} is the first result to analyze structured signal recovery from block-sparse corruption.

Let $m$ denote the number of disjoint measurement blocks and $k= n/m$ represent the common block size.
Suppose that no more than $s$ measurement blocks have been corrupted and that the identities of the corrupted blocks are unknown.
We consider a ``dense'' corruption regime where the proportion of corruptions $\gamma=\frac{s}{m}$ may lie anywhere in $(0,1)$.
The following corollary bounds the number of measurements needed for exact or stable recovery, using the penalized program
\begin{align}\label{eq:pen-gen-block}
\min_{x,v}\left\{\snorm{x} + \lambda\norm{v}_{\ell_1/\ell_2}:\norm{y-(\Phi x+ v)}_2\leq \delta\right\}
\end{align}
with $\lambda={\mu_k(1-\gamma)/}{\subsig{t}}$ for any choice of $\subsig{t} \geq 0$.
In particular, when $\delta = 0$, we find that
\[
\O{\frac{\sqgdist{\subsig{t}\cdot\partial\snorm{x^{\star}}}}{(1-\gamma)^4}}
\] 
measurements suffice for exact recovery of a structured signal $x^{\star}$.

\begin{corollary}[Penalized recovery, dense block or entrywise corruption] \label{cor:pen-gen-block}
Fix any $\subsig{t}\geq 0$, and suppose that $v^{\star}$ exhibits block-sparse structure with blocks of size $k$.
If the fraction of nonzero blocks of $v^{\star}$ is at most $\gamma$,
$(\widehat{x},\widehat{v})$ solves \eqref{eq:pen-gen-block}
with the penalty parameter $\lambda=\frac{\mu_k(1-\gamma)}{\subsig{t}}$,
and $\sqrt{n}$ is at least as large as
\[
\frac{2\gdist{\subsig{t}\cdot\partial\snorm{x^{\star}}} + \sqrt{2\log(\nicefrac{1}{\beta})} +3\sqrt{2\pi}+ 1 +\frac{1}{\sqrt{2\pi}}}{\left(\alpha_k\left(1-\gamma\right)^2 - \eps\right)_+}
\]
for $\alpha_k = 1-\sqrt{1-\mu_k^2/k}$, then with probability at least $1-\beta$,
\[\sqrt{\norm{\widehat{x}-x^{\star}}^2_2+\norm{\widehat{v}-v^{\star}}^2_2}\leq \frac{2\delta}{ \eps}\;.\]
\end{corollary}
\begin{remark}
In the noiseless setting, where $\delta = 0$, \corref{pen-gen-block} entails \emph{exact} recovery with probability $1-\beta$ whenever $\sqrt{n}$ is at least as large as
\[
\frac{2\gdist{\subsig{t}\cdot\partial\snorm{x^{\star}}} + \sqrt{2\log(\nicefrac{1}{\beta})} +3\sqrt{2\pi}+ 1 +\frac{1}{\sqrt{2\pi}}}{\alpha_k\left(1-\gamma\right)^2}.
\]
\end{remark}
In the important special case where $x^{\star}$ is a sparse vector with at most $s$ non-zero entries and $v^{\star}$ is entrywise sparse ($k=1$), 
the $\ell_1$ distance bound in~\tabref{WidthBounds} and \corref{pen-gen-block} together imply that
\[
\O{\frac{s\log\left({p}/{s}\right) + s }{(1-\gamma)^4}}
\qtext{or}
\O{\frac{s\log\left({p}/{s}\right) + s }{(\alpha_1(1-\gamma)^2-\eps)_+^2}}
\] measurements suffice to recover $x^{\star}$ exactly or stably, respectively, with high probability, using the convex program
\begin{align} \label{eqn:pen-sparse-sparse}
\min_{x,v}\left\{\norm{x}_1 + \lambda\norm{v}_1:\norm{y-(\Phi x+ v)}_2\leq \delta\right\},
\end{align}
with $\lambda = (1-\gamma)/\sqrt{\pi\log(p/s)}$.
This result is particularly well suited to recovery from dense entrywise corruption and accords with the best sparse recovery from sparse corruption results in the literature. 
For example, \citet[Thm.~1.1]{Li13} establishes stable recovery via \eqref{eqn:pen-sparse-sparse} using the penalty parameter $\lambda = 1/\sqrt{\log(p/n)+1}$, whenever the signal and corruption sparsity levels satisfy
\begin{equation*}\label{eqn:Li_assumptions}
\norm{x^{\star}}_0 <  C \cdot \frac{n}{\log(p/n)+1}\qtext{ and }\norm{v^{\star}}_0\leq{C}\cdot n\;,
\end{equation*}
for an unspecified universal constant $C > 0$. Our result admits a guarantee of this form, while providing an explicit trade-off between the sparsity levels of the signal and the corruption. 

Specifically, \corref{pen-gen-block} yields a stable recovery result (derived in \appref{LiCompare}) under the conditions
\begin{equation} \label{eqn:LiCompare}
\norm{x^{\star}}_0 \leq  s_\gamma := \left\lfloor\frac{A_\gamma n}{2\log(\frac{p}{A_\gamma n})+3/2}\right\rfloor\text{ and }\norm{v^{\star}}_0\leq  n\cdot\gamma\;,
\end{equation}
for any corruption proportion bound $\gamma<1$ and 
\[A_\gamma:=(\alpha_1\left(1-\gamma\right)^2 - \eps)_+^2/144\;.\]
(Recall that $\alpha_1 = 1-\sqrt{1-\nicefrac{2}{\pi}} \approx 0.4$.)
The associated setting of the penalty parameter $\lambda = (1-\gamma)/{\sqrt{\pi\log(\nicefrac{p}{s_\gamma})}}$ is parameterized only by $\gamma$ and, like \citeauthor{Li13}'s result, requires no prior knowledge of the signal sparsity level $\norm{x^{\star}}_0$.
Our setting has the added advantage of allowing the user to specify an arbitrary upper bound on the corruption level.

It should be noted  that \citet{Li13} addresses the adversarial setting in which $x^{\star}$ and $v^{\star}$ may be selected given knowledge of the measurement matrix $\Phi$.
Our work can be adapted to the adversarial setting by measuring the Gaussian complexity of the $\ell_1$ unit ball in place of the associated tangent cones; we save such an extension for future work.

\subsection{Penalized recovery when signal and corruption are both extremely sparse}
\corref{pen-gen-block} is well-suited to recovery from moderate or frequent corruption, but
in the extreme sparsity setting, where the corruption proportion $\gamma$ is nearly zero, the result can be improved with a different choice of penalty parameter $\lambda$; the reason for this distinction is that, when choosing $\subcor{t}$, it is more advantageous to use the first distance bound that is given in \tabref{DistBounds} for the block-wise sparse setting  when the corruptions are extremely sparse, rather than using the second distance bound in this table as for the high-corruption setting. A general result for the extremely sparse case could be attained with a similar analysis as in \corref{pen-gen-block}; we do not state such a result here but instead give a specific example where both the signal and corruption exhibit extreme element-wise sparsity. 
 
 In the following corollary, we show that 
 \[\O{(\subsig{s}+\subcor{s})\left(\log\left(\frac{p+n}{\subsig{s}+\subcor{s}}\right)+1\right)}\]
measurements suffice to ensure recovery even without any knowledge of the sparsity levels $\subsig{s}$ and $\subcor{s}$, by selecting the penalty parameter $\lambda=1$.  
 
 \begin{corollary}[Extremely sparse signal and corruption]\label{cor:pen-extreme-sparsity}
 Suppose that $x^{\star}$ is supported on at most $\subsig{s}$ of $p$ entries, 
and that $v^{\star}$ is supported on at most $\subcor{s}$ of $n$ entries.
If $(\widehat{x},\widehat{v})$ solves
\begin{equation*}\label{eqn:sparse-sparse}\min_{x,v}\left\{\norm{x}_1+\norm{v}_1:\norm{y-(\Phi x+v)}_2\leq \delta\right\}\end{equation*}
and the number of measurements $n$ satisfies
\begin{multline*}\mu_n -\eps\sqrt{n}\geq \tau\coloneqq\\ 3\sqrt{(\subsig{s}+\subcor{s})\cdot\left(2\log\left(\frac{p+n}{\subsig{s}+\subcor{s}}\right)+\frac{3}{2}\right)}\\+ 3\sqrt{2\pi}+\frac{1}{\sqrt{2}}+\frac{1}{\sqrt{2\pi}}\;,\end{multline*}
then with probability at least $1-e^{-(\mu_n-\tau-\eps\sqrt{n})^2/2}$,
\begin{align*} \label{eqn:pen-block-sparse-result}
\sqrt{\norm{\widehat{x}-x^{\star}}^2_2+\norm{\widehat{v}-v^{\star}}^2_2}\leq \frac{2\delta}{ \eps}\;.
\end{align*}
 \end{corollary}

When the noise level equals zero, we recover a result due to \citet{LaskaDaBa09}, which establishes high probability exact recovery of a sparse signal from sparsely corrupted measurements via the convex program
\begin{align*} 
\min_{x,v}\left\{\norm{x}_{1} +\norm{v}_{1}:y=\Phi x+ v\right\}.
\end{align*}
Their analysis requires $n \geq C(\subsig{s}+\subcor{s})\log(\frac{p+n}{\subsig{s}+\subcor{s}})$ measurements for recovery, where $C$ is an unspecified constant.
In comparison, the present analysis features small, explicit constants, provides for stable recovery in the presence of unstructured noise, and follows from a more general treatment of signal and corruption structure.

We remark that a similar result to \corref{pen-extreme-sparsity} can be stated for block-sparse signals with $\subsig{m}$ blocks of size $\subsig{k}$ and block-sparse corruption vectors with $\subcor{m}$ blocks of size $\subcor{k}$.
In this case,
\[\O{(\subsig{s}+\subcor{s})\left(\log\left(\frac{\subsig{m}+\subcor{m}}{\subsig{s}+\subcor{s}}\right)+\subsig{k}+\subcor{k}\right)}\]
measurements suffice for high probability recovery. We omit the details here.

\subsection{Additional related work}
\label{sec:related-work}
We conclude this section with a summary of some existing results in the vicinity of this work.

\subsubsection{Sparse signal recovery from sparse corruption}
The specific problem of recovering a sparse vector $x^{\star}\in\R^p$ from linear measurements with sparse corruption $v^{\star}\in\R^n$
has been analyzed under a variety of different modeling assumptions.
For example, \citet{WrightMa10} study exact recovery under a \emph{cross-and-bouquet} model: $y=Ax^{\star} +v^{\star}$ for 
\[A^{(i)}\iidsim N(\mu,\frac{\nu^2}{n}\mathbf{I}_p)\text{, }\norm{\mu}_2=1,\norm{\mu}_{\infty}\leq \frac{C}{\sqrt{n}}\;,\]
where $A^{(i)}$ is the $i$th column of $A$, and $v^{\star}$ has uniformly random signs. 
\citet{NguyenTr13} establish stable or exact recovery of $(x^{\star},v^{\star})$ from $y=Ax^{\star} +v^{\star}+z$, where $A$ has columns sampled uniformly from an orthonormal  matrix, $x^{\star}$ has uniformly random signs, $v^{\star}$  has uniformly distributed support, and $z$ is a bounded noise vector.
\citet{PopeBrSt13}  analyze the exact recovery of $x^{\star}$ and $v^{\star}$ from $y= Ax^{\star}+Bv^{\star}$, where $A$ and $B$ are known.
The authors require uniform randomness in the support set of $x^{\star}$ or $v^{\star}$ and rely on certain incoherence properties of $A$ and $B$.  
In contrast to these works, we analyze the recovery of deterministic signal and corruption vectors from Gaussian measurements,
derive small, explicit constants, and generalize to arbitrary structured signals and structured corruption vectors.

\subsubsection{Structured signal recovery from structured corruption}
\citet{HegdeBa12} analyze a nonconvex procedure for deconvolving a pair of signals lying on incoherent manifolds when the measurement matrix satisfies a restricted isometry property.
Here, we focus on convex procedures for which global minima can be found in polynomial time.

In addition to their work on the binary signal plus sparse corruption problem, \citet{McCoyTr12} give recovery results for arbitrary structured signals and structured corruptions, in the setting where $n=p$, observations are noiseless ($\delta=0$ in our notation), and either $\snorm{x^{\star}}$ or $\cnorm{v^{\star}}$ is known exactly. 
In the same setting, the recent work of \citet{AmelunxenLoMcTr13} independently relates the notion of Gaussian (squared) distance to  the probability of signal recovery.  
These works also establish the sharpness of their recovery results, by demonstrating a decay in success probability whenever the sample size falls below an appropriate threshold.
In contrast, our work does not discuss high probability failure bounds but allows us to consider $n$ smaller or larger than $p$, to handle stable recovery in the presence of noise, and to use penalized rather than constrained optimization where practical. 

\section{Experimental Evaluation}
\label{sec:experiments}
In this section, we verify and complement the theoretical results of \secref{theory} with a series of synthetic corrupted sensing experiments.
We present both constrained and penalized recovery experiments for several types of structured signals and corruptions in the presence or absence of noise. 
Our goals are twofold: first, to test the agreement between our constrained recovery theory and empirial recovery behavior and, second, to evaluate the utility of the penalty parameter settings suggested in \thmref{gen-pen-recovery} when using penalized recovery.
In each experiment, we employ the CVX Matlab package~\citep{cvx,GrantBo08} to specify and solve our convex recovery programs.

\subsection{Phase transitions from constrained recovery} \label{sec:phase-con}

We begin by investigating the empirical behavior of the constrained recovery program~\eqref{eqn:sig-bound-prob} when the noise level $\delta = 0$ and the norm of the true signal $\snorm{x^{\star}}$ are known exactly. 
 In this setting, our constrained recovery result (\thmref{gen-con-recovery}) guarantees exact recovery with some probability once $\mu_n^2\ (\approx n)$ exceeds $\sqgcomplexity{\stcone\cap\twoball{p}} + \sqgcomplexity{\ctcone\cap\twoball{n}}$.  
 Using new and existing bounds on the relevant Gaussian complexities, our aim is to determine how conservative this theoretical recovery guarantee is in practice. 
To this end, we will apply our convex recovery procedures to synthetic problem instances and compare the empirical probability of successful recovery to the recovery behavior predicted by our theory.
We will see that our theoretical recovery thresholds closely align with observed phase transitions in three different settings: 
binary signal recovery from sparse corruption, sparse signal recovery from sparse corruption, and  sparse signal recovery from block-wise sparse corruption.

\paragraph*{Binary signal, sparse corruption}
We first consider the secure and robust communications protocol discussed in \secref{con-binary-sparse}, where we aim to recover a binary message
from noiseless ($\delta = 0$) but sparsely corrupted Gaussian measurements.
Fixing the message length $p = 1000$, we vary the number of measurements $n \in [750, 1250]$ and the number of corruptions $\subcor{s} \in [50, 350]$ and
perform the following experiment 10 times for each $(n,\subcor{s})$ pair:
\begin{enumerate}
\item Sample a binary vector $x^{\star}$ uniformly from $\{\pm 1\}^p$.
\item Draw a Gaussian matrix $\Phi \in \R^{n \times p}$ with independent $N(0,1/n)$ entries.
\item Generate a corruption vector $v^{\star}$ with $\subcor{s}$ independent standard normal entries and $n-\subcor{s}$ entries set to $0$.
\item Solve the constrained optimization problem~\eqref{eqn:con-binary-sparse-prob} with $y = \Phi x^{\star} + v^{\star}$ and $\delta = 0$.
\item Declare success if $\norm{\widehat{x}-x^{\star}}_2/\norm{x^{\star}}_2 < 10^{-3}$.
\end{enumerate}

\figref{con-binary-sparse} reports the empirical probability of success for each setting of $(n,\subcor{s})$ averaged over the 10 iterations.
 To compare these empirical results with our theory, we overlay the theoretical recovery thresholds
\[
\mu_n^2 = \sqgcomplexity{\stcone\cap\twoball{p}} + \sqgcomplexity{\ctcone\cap\twoball{n}}
\]
where the Gaussian complexity for the binary signal is estimated with the bound given in \tabref{WidthBounds}, while the complexity for the sparse corruption is given by the optimal squared distance bound for the $\ell_1$ norm, given in \eqref{eqn:SparseDistSquared}.
This threshold shows the parameter regions where our theory gives any positive probability of successful recovery (we have ignored the small additive constants, which are likely artifacts of the proof and are negligible for large $n$). 
We find that our theory provides a practicable phase transition that reflects empirical behavior.

\begin{figure*} 
\centering
\includegraphics[scale=0.78]{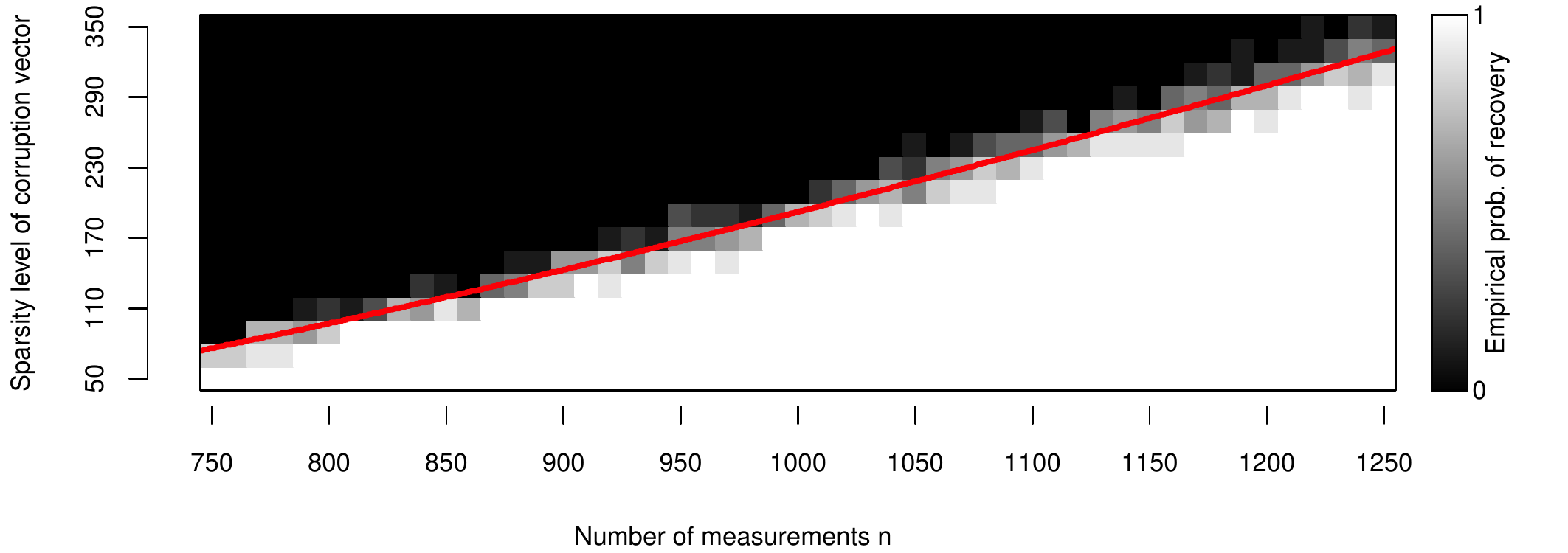}
\caption{Phase transition for binary signal recovery under sparse corruption, using constrained recovery (see \secref{phase-con}). The red curve plots the recovery threshold predicted by \thmref{gen-con-recovery} (ignoring the small additive constants that we believe are artifacts of the proof technique), where Gaussian squared complexity for the binary signal is given in \tabref{WidthBounds}, and Gaussian squared complexity for the sparse corruption is estimated by minimizing \eqref{eqn:SparseDistSquared}.}
\label{fig:con-binary-sparse}
\end{figure*}

\paragraph*{Sparse signal, sparse corruption}
We next consider the problem of recovering a sparse signal vector of known $\ell_1$ norm from sparsely corrupted Gaussian measurements.
We fix the signal length and sample size $p = n = 1000$, and vary the sparsity levels $(\subsig{s},\subcor{s}) \in [1,860]^2$. We
perform the following experiment 10 times for each $(\subsig{s},\subcor{s})$ pair:
\begin{enumerate}
\item Generate a signal vector $x^{\star}$ with $\subsig{s}$ independent standard normal entries and $p-\subsig{s}$ entries set to $0$.
\item Generate a corruption vector $v^{\star}$ with $\subcor{s}$ independent standard normal entries and $n-\subcor{s}$ entries set to $0$.
\item Draw a Gaussian matrix $\Phi \in \R^{n \times p}$ with independent $N(0,1/n)$ entries.
\item Solve the following constrained optimization problem with $y = \Phi x^{\star} + v^{\star}$:
\begin{multline*}
(\widehat{x},\widehat{v}) \in \argmin_{x,v}\big\{\norm{v}_1:\\\norm{x}_1\leq \norm{x^{\star}}_1, y = \Phi x + v \big\}.
\end{multline*}
\item Declare success if $\norm{\widehat{x}-x^{\star}}_2/\norm{x^{\star}}_2 < 10^{-3}$.
\end{enumerate}
\figref{con-sparse-sparse} reports the average empirical probability of success for each setting of $(\subsig{s},\subcor{s})$. We again overlay the theoretical recovery threshold suggested by \thmref{gen-con-recovery}, where 
the Gaussian complexities for both signal and corruption are estimated with the tight upper bounds given by the optimal squared distance bound for the $\ell_1$ norm, given in \eqref{eqn:SparseDistSquared}. This theoretical recovery threshold reflects the observed empirical phase transition accurately.

\begin{figure*} 
\centering
\includegraphics[scale=0.6]{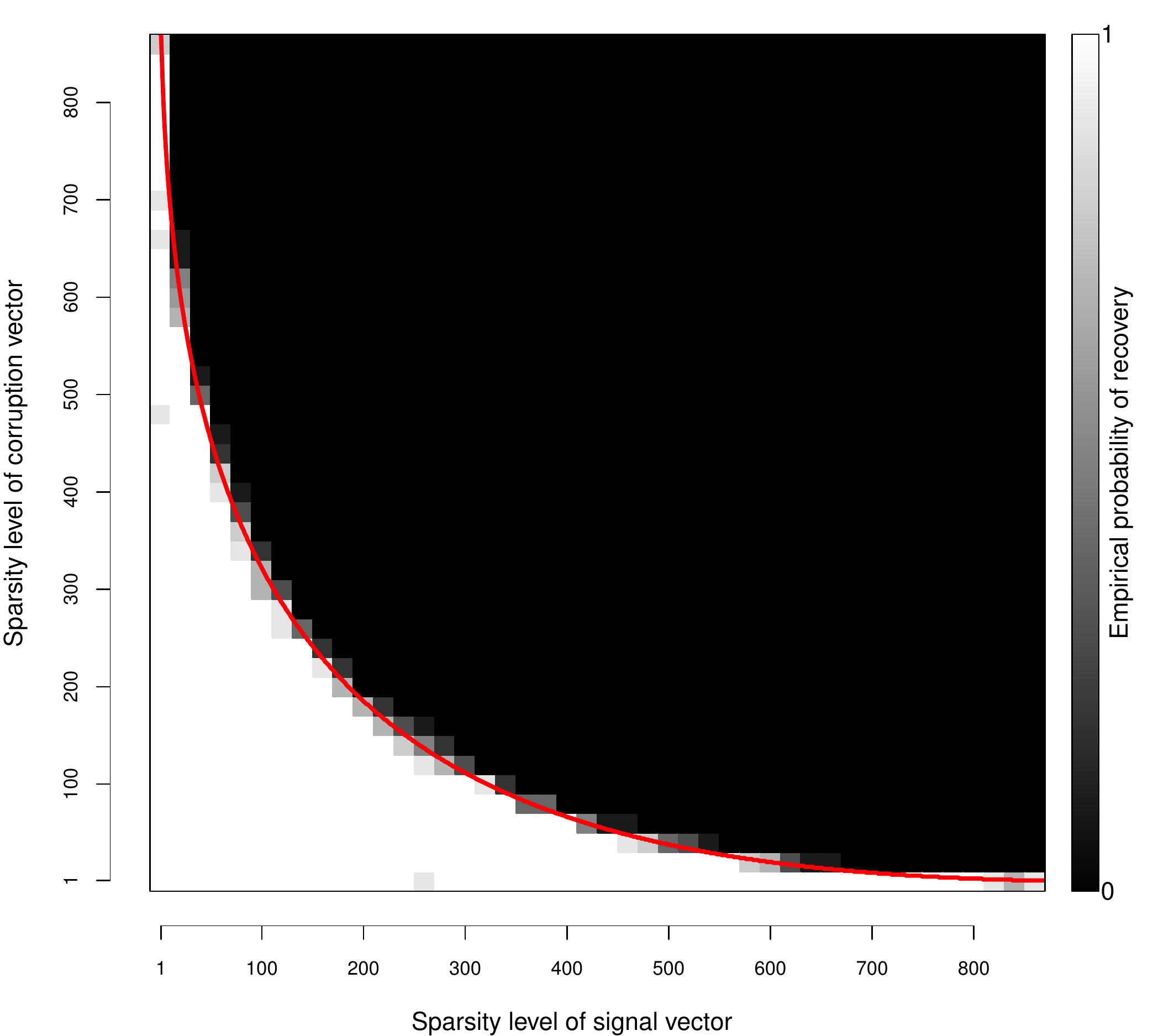}
\caption{Phase transition for sparse signal recovery under sparse corruption, using constrained recovery (see \secref{phase-con}). The red curve plots the recovery threshold predicted by \thmref{gen-con-recovery} (ignoring the small additive constants that we believe are artifacts of the proof technique), where Gaussian complexity is estimated by minimizing \eqref{eqn:SparseDistSquared}.}
\label{fig:con-sparse-sparse}
\end{figure*}

\paragraph*{Sparse signal, block-sparse corruption}
To further test the resilience of our theory to a change in structure type, we next
consider the problem of recovering a sparse signal vector of known $\ell_1$ norm from measurements with block-sparse corruptions. 
We fix $p = n = 1000$ and partition the indices of $v^{\star}$ into $m = 100$ blocks of size $k=10$. 
We test sparsity and block-sparsity levels $(\subsig{s},\subcor{s}) \in [1,860]\times[1,86]$. 
For each $(\subsig{s},\subcor{s})$ pair, we follow the experimental setup of the previous experiment with two substitutions:
we generate $v^{\star}$ with $\subcor{s}$ blocks of independent standard normal entries and $m - \subcor{s}$ blocks with all entries set to zero, and
we solve the constrained optimization problem
\begin{multline*}
(\widehat{x},\widehat{v}) \in \argmin_{x,v}\big\{\norm{v}_{\ell_1/\ell_2}:\\\norm{x}_1\leq \norm{x^{\star}}_1, y = \Phi x + v \big\}.
\end{multline*}

Results from this simulation are displayed in \figref{con-sparse-groupsparse}, with the overlaid threshold coming from the optimal squared distance bound for the $\ell_1$ norm, given in \eqref{eqn:SparseDistSquared}, and for the $\ell_1/\ell_2$ norm, as follows (see \eqref{eqn:groupsparse-exact} in \appref{block} for the derivation):
\begin{multline*}
\min_{t\geq0}\EE{\inf_{w\in t\cdot\partial\norm{v^{\star}}_{\ell_1/\ell_2}} \norm{g - w}_2^2} \\
= \min_{t\geq0} \subcor{s}(t^2+k) + \frac{2^{1-k/2}(m-\subcor{s})}{\Gamma(k/2)} \cdot\\\int_t^\infty (c-t)^2c^{k-1}e^{-c^2/2}\ \diff{c}.
\end{multline*}
Again, we see that our theory accurately matches the observed phase transition of successful recovery. For comparison, we also plot the theoretical recovery threshold from the previous experiment (sparse signal and sparse corruption), which would be the recovery threshold for optimization using the $\ell_1$ norm on both signal and corruption, rather than leveraging the additional block structure via the $\ell_1/\ell_2$ norm.

\begin{figure*}
\centering
\includegraphics[scale=0.6]{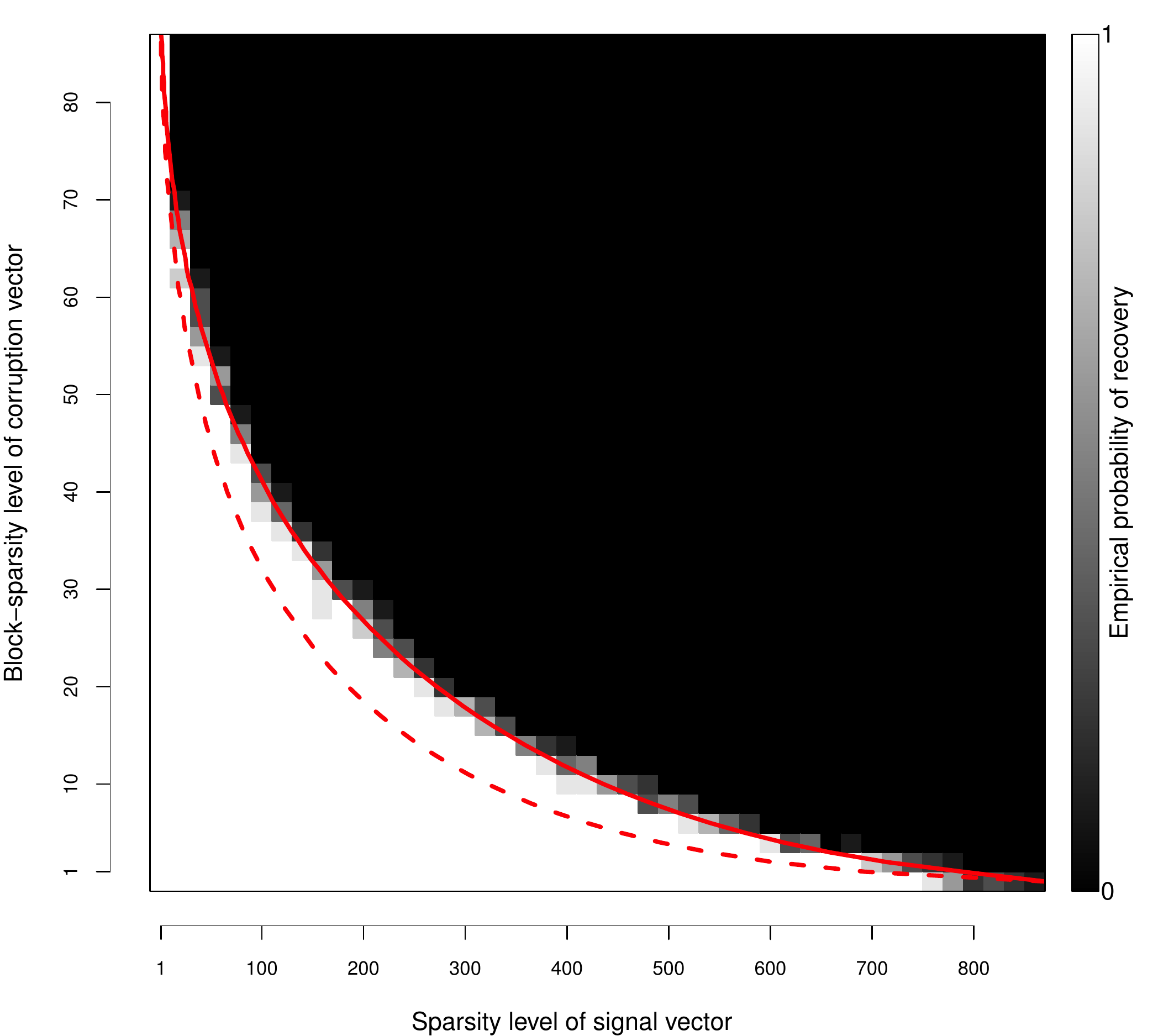}
\caption{Phase transition for sparse signal recovery under block-wise sparse corruption, using constrained recovery (see \secref{phase-con}). The solid red curve plots the recovery threshold predicted by \thmref{gen-con-recovery} (ignoring the small additive constants that we believe are artifacts of the proof technique), where Gaussian complexity is estimated analogously to the sparse signal plus sparse corruption problem. The dashed red curve plots the recovery threshold that would be predicted if the $\ell_1$ norm, rather than the $\ell_1/\ell_2$ norm, were used on the corruption term (the same curve that appears in \figref{con-sparse-sparse}); this shows the benefit of leveraging the block structure in the corruption.}
\label{fig:con-sparse-groupsparse}
\end{figure*}

\subsection{Phase transitions from penalized recovery}
\label{sec:phase-pen}

\begin{figure*}
        \centering
        \begin{subfigure}[b]{0.49\textwidth}
                \centering
			\includegraphics[scale=0.5]{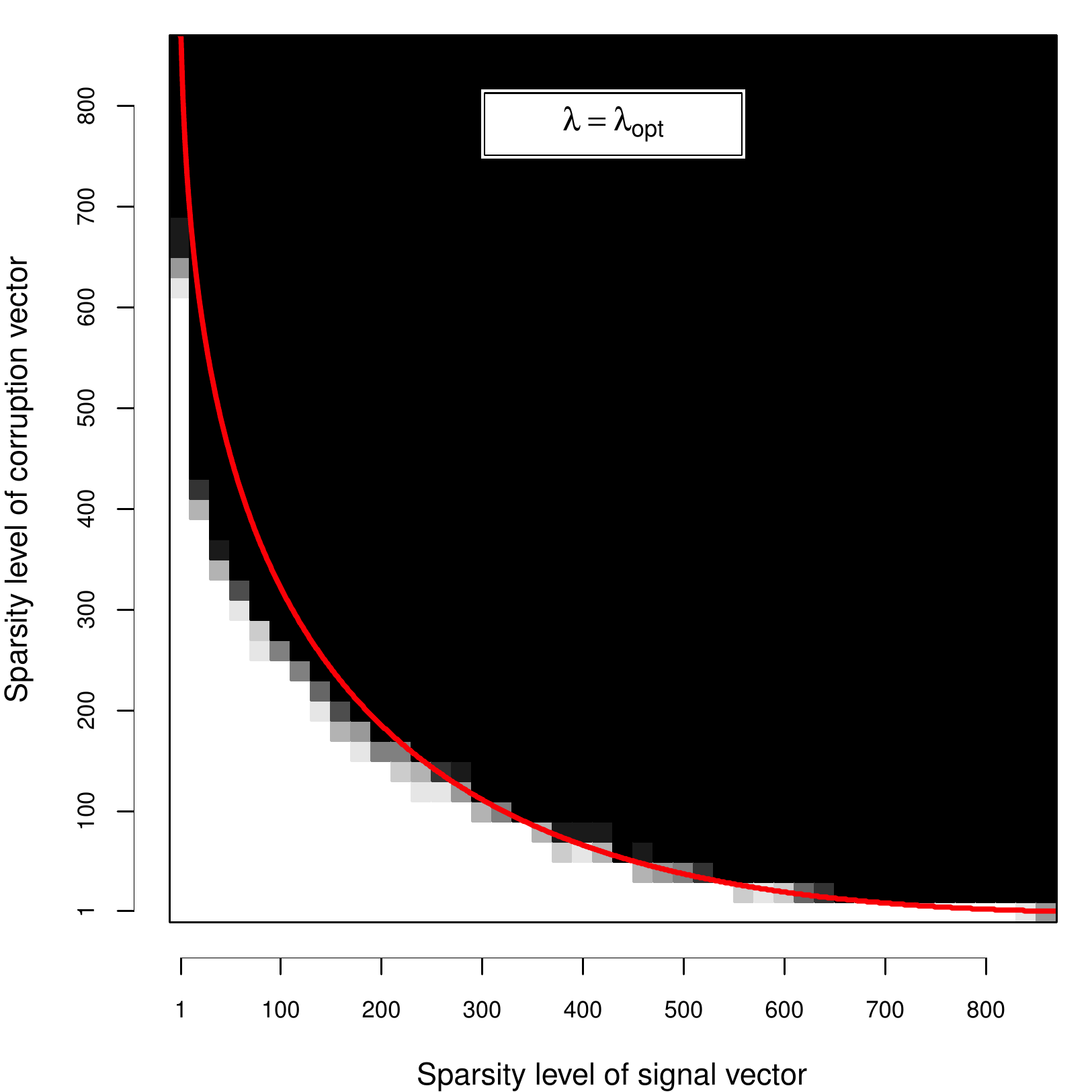}
        \end{subfigure}
        ~ 
        \begin{subfigure}[b]{0.49\textwidth}
                \centering
			\includegraphics[scale=0.5]{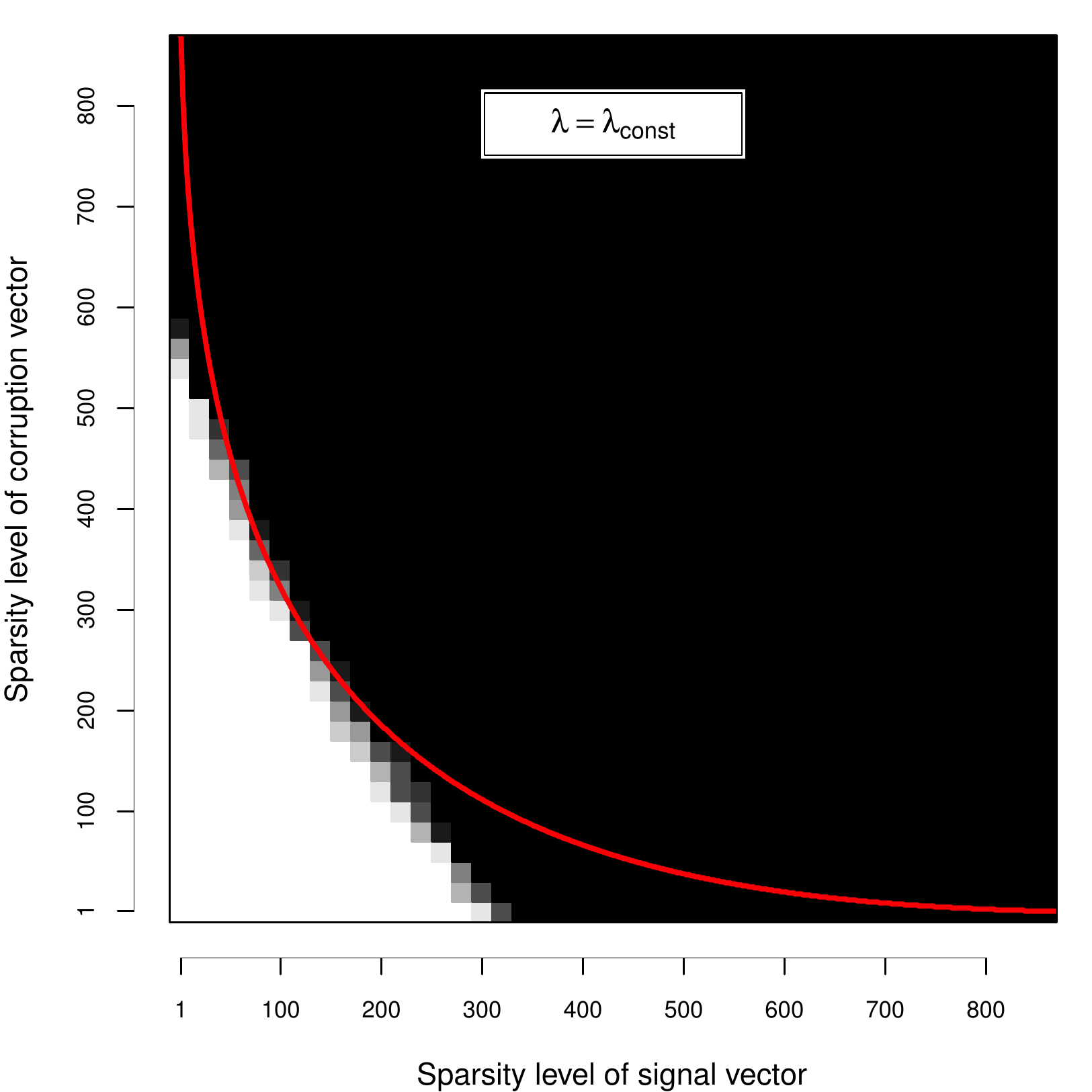}
        \end{subfigure}
        
        \begin{subfigure}[b]{0.49\textwidth}
                \centering
                \includegraphics[scale=0.5]{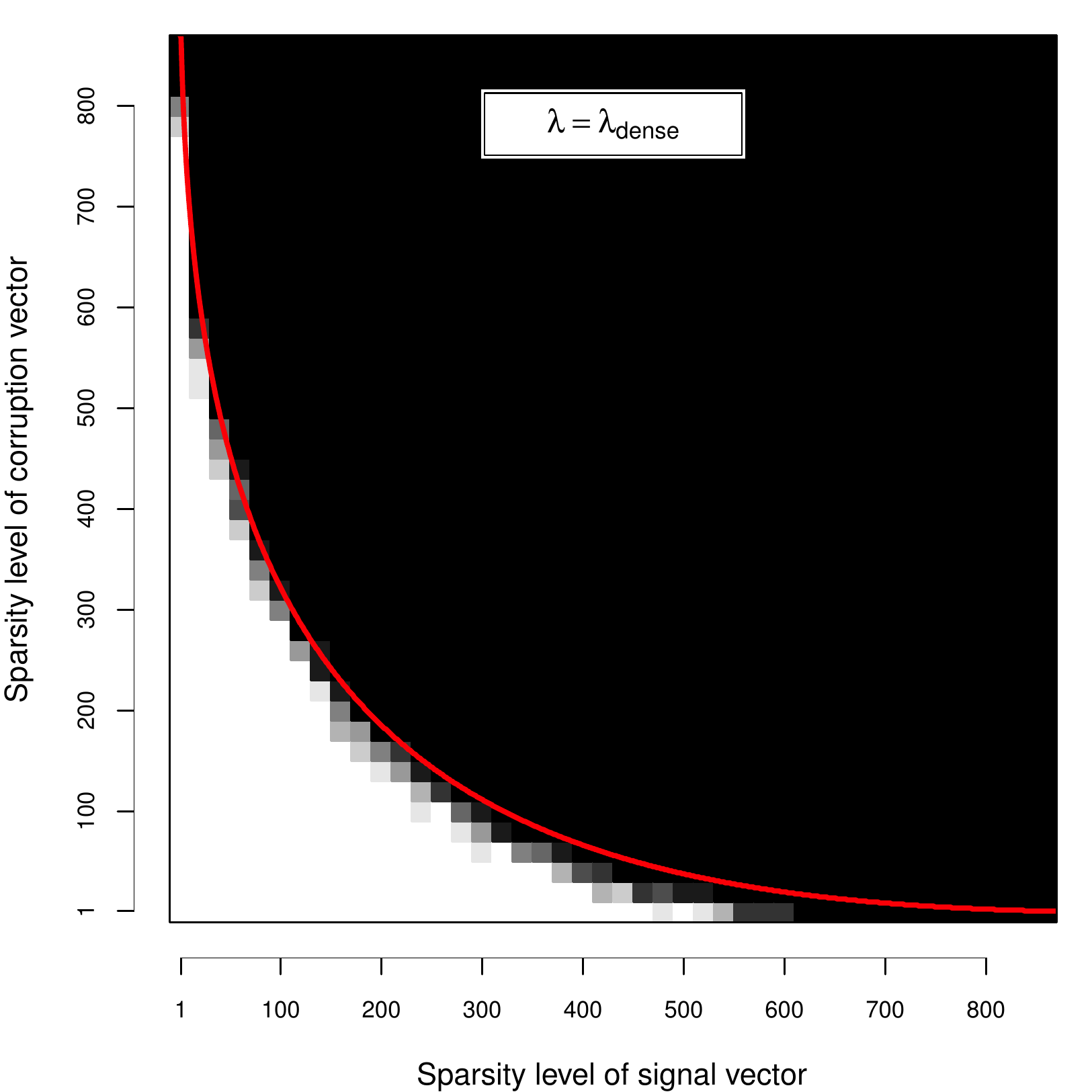}
        \end{subfigure}
        ~
        \begin{subfigure}[b]{0.49\textwidth}
                \centering
                \includegraphics[scale=0.5]{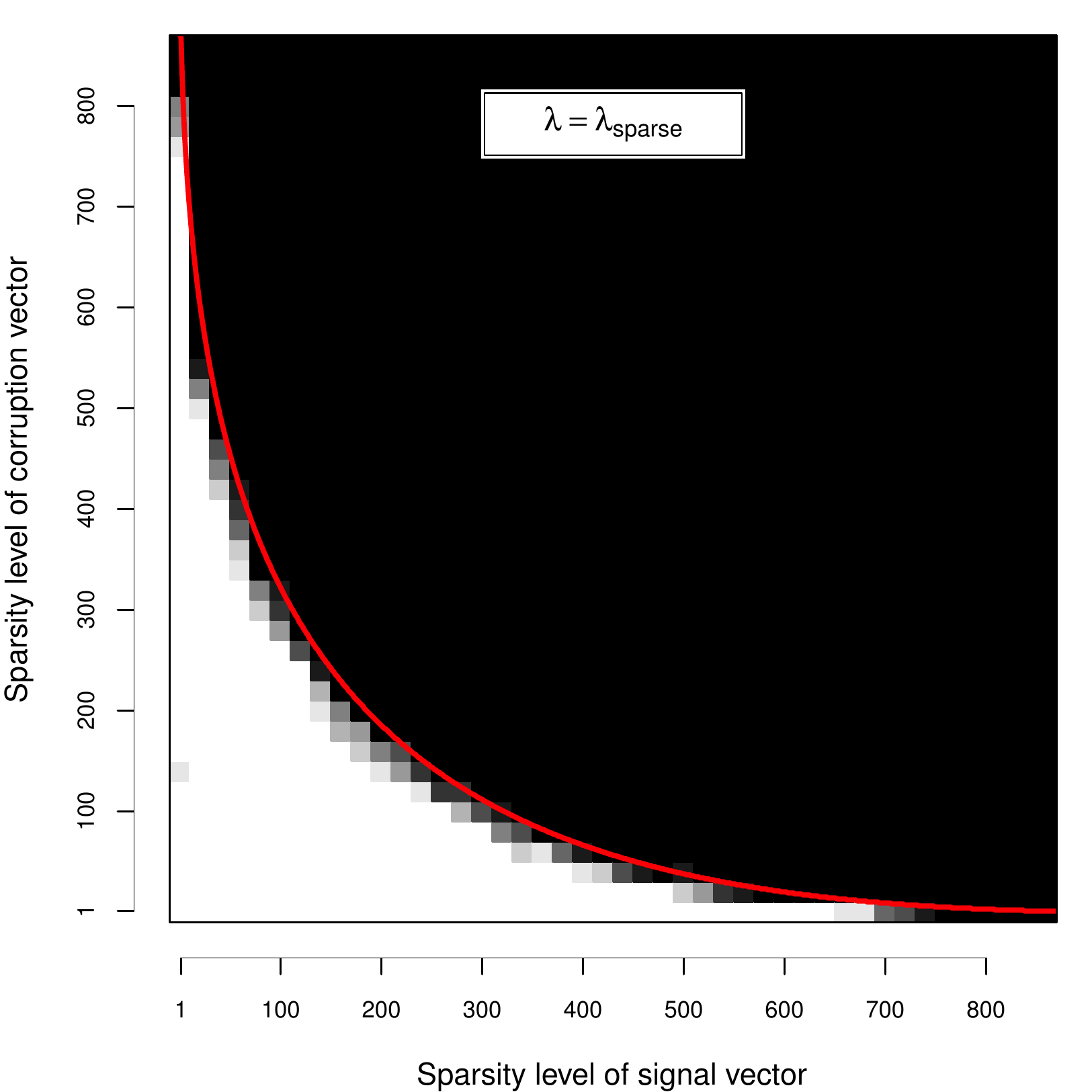}
        \end{subfigure}
        \caption{Phase transitions for sparse signal recovery under sparse corruption, using penalized recovery with varied settings of the penalty parameter $\lambda$ (see \secref{phase-pen}). Grayscale representation of success probability is same as in \figref{con-sparse-sparse}. The red curve in each figure is the theoretical recovery threshold predicted by \thmref{gen-con-recovery} for the constrained recovery problem (same curve as in \figref{con-sparse-sparse}); we display this curve to show that signal recovery is nearly as good as in the constrained case, but without any prior knowledge of $\norm{x^\star}_1$ or $\norm{v^\star}_1$.
}        \label{fig:pen-sparse-sparse}
\end{figure*}

Next, we will consider the penalized recovery program~\eqref{eqn:pen-prob} in the noiseless setting where neither $\snorm{x^{\star}}$ nor $\cnorm{v^{\star}}$ is known a priori. We focus on a sparse signal $x^{\star}$ plus sparse corruption $v^{\star}$ model---note that for this type of structure, in practice neither $\snorm{x^{\star}}=\norm{x^{\star}}_1$ nor $\cnorm{v^{\star}}=\norm{v^{\star}}_1$ will be known exactly a priori, and therefore penalized recovery is often more practical than constrained recovery.

Since constrained recovery is guaranteed to be successful at a sample size of approximately $\sqgcomplexity{\stcone\cap\twoball{p}} + \sqgcomplexity{\ctcone\cap\twoball{n}}$, the known correspondence between constrained and penalized optimization implies that for some (unknown) value of $\lambda$, the penalized recovery program should also yield exact recovery at this sample size. The primary difficulty in the penalized setting, however, lies in choosing a penalty parameter that leads to good recovery behavior without prior knowledge of the signal or corruption norm. 
Our penalized recovery result, \thmref{gen-pen-recovery}, suggest a simple strategy for picking the penalty: set $\lambda = \subcor{t}/\subsig{t}$ where $\subsig{t}$ is a scaling that leads to a small bound on the signal Gaussian distance 
$\gdist{\subsig{t}\cdot\partial\norm{x^{\star}}_1}$ and $\subcor{t}$ is an analogous scaling for a corruption Gaussian distance.
We will test four settings of the penalty parameter,  three of which depend on the signal and corruption sparsity levels $\subsig{s}$ and $\subcor{s}$:
\begin{enumerate}
	\item $\lambda_{\text{sparse}} = \sqrt{\log(n/\subcor{s})}/\sqrt{\log(p/\subsig{s})}$, based on the closed-form bound \eqref{eqn:SparseConeBound_old} that is nearly optimal for highly sparse vectors,
	\item $\lambda_{\text{dense}} = (1-\subcor{s}/n)/(1-\subsig{s}/p)$, based on the closed-form bound \eqref{eqn:SparseConeBound_new} that is nearly optimal for all but the sparsest vectors, 
	\item $\lambda_{\text{opt}}$, which chooses $\subsig{t}$ and $\subcor{t}$ to minimize the respective expected squared distances exactly via \eqref{eqn:SparseDistSquared}, and
\item $\lambda_{\text{const}} = 1$, for which we expect recovery when both $\subsig{s}$ and $\subcor{s}$ are sufficiently sparse (see \corref{pen-extreme-sparsity}). 
\end{enumerate}

For direct comparison with the constrained recovery case, we generate sparse signal, sparse corruption pairs 
and noiseless Gaussian measurements with $p = n = 1000$, precisely as in \secref{phase-con}.
To recover each $(x^{\star},v^{\star})$ pair, we solve the penalized optimization problem~\eqref{eqn:pen-sparse-sparse} with each penalty parameter setting and declare success if $\norm{\widehat{x} - x^{\star}}_2/\norm{x^{\star}}_2 < 10^{-3}$.
\figref{pen-sparse-sparse} displays the empirical probability of success as the signal and corruption sparsity levels $\subsig{s}$ and $\subcor{s}$ vary.  
For reference, the theoretical recovery curve of \figref{con-sparse-sparse} has been overlaid.
Remarkably, setting $\lambda$ to equal any of $\lambda_{\text{sparse}}$, $\lambda_{\text{dense}}$, or $\lambda_{\text{opt}}$ yields empirical performance nearly as good that obtained in the constrained setting (\figref{con-sparse-sparse}) but without knowledge of $\norm{x^{\star}}_1$ or $\norm{v^{\star}}_1$. 
In other words, although \thmref{gen-pen-recovery} requires a larger number of measurements than \thmref{gen-con-recovery} to guarantee success, the penalized recovery program offers nearly the same phase transition as the constrained program in practice.
Moreover, as predicted in \corref{pen-extreme-sparsity}, when estimates of the sparsity levels $\subsig{s}$ and $\subcor{s}$ are unavailable, the setting $\lambda=\lambda_{\text{const}}=1$ yields high probability recovery, provided that neither signal nor corruption is overly dense. 
\subsection{Stable recovery error}

\begin{figure*}
\centering
\includegraphics[scale=0.42]{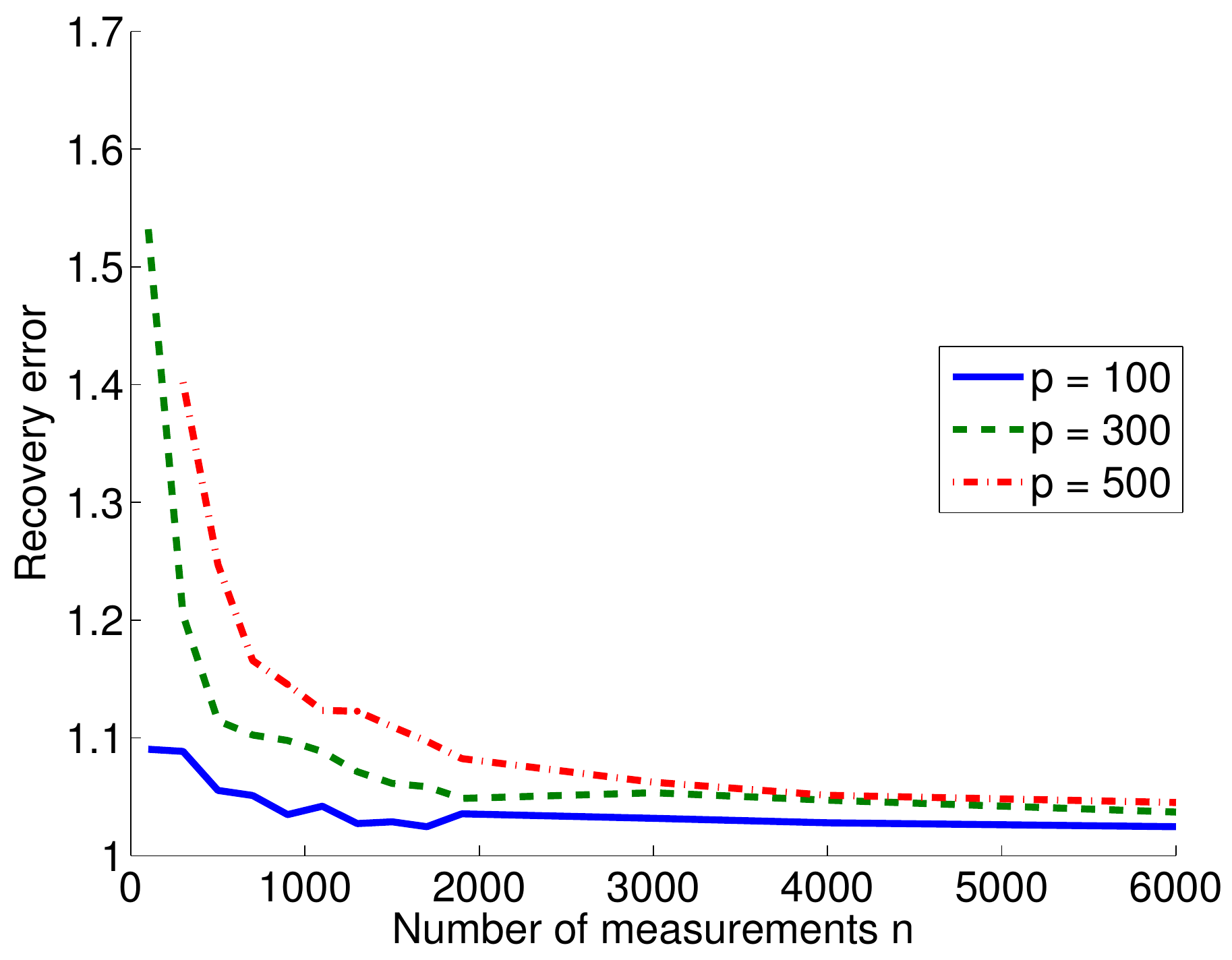}
\includegraphics[scale=0.42]{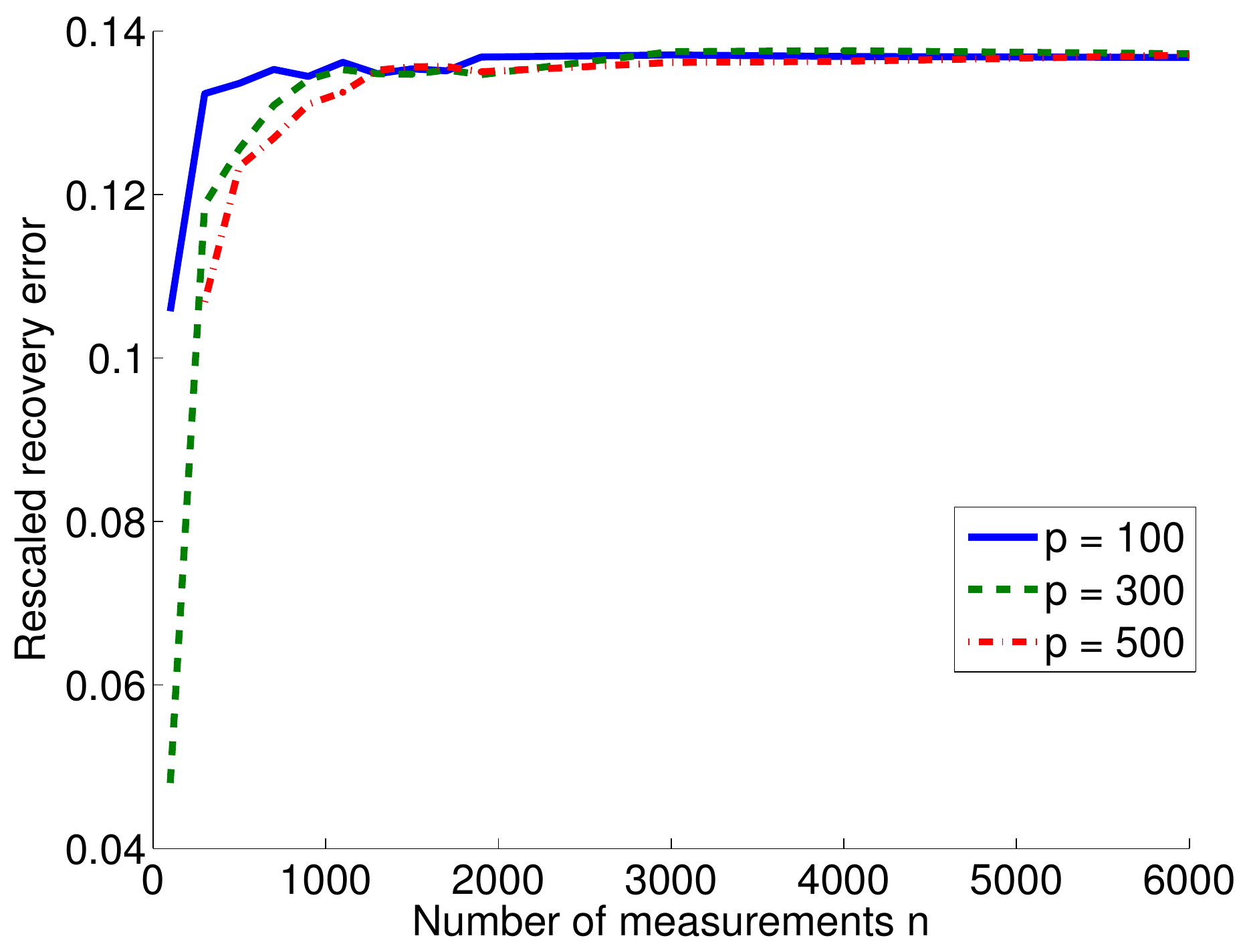}
\caption{(Left) Recovery error $\sqrt{\norm{\widehat{x}-x^{\star}}_2^2 + \norm{\widehat{v}-v^{\star}}_2^2}$ for sparse signal recovery ($\subsig{s} = 0.01p$) under noise and sparse corruption ($\subcor{s}=0.4n$), using constrained recovery.
(Right) Same recovery error rescaled by the optimal squared distance estimate \eqref{eqn:SparseDistSquared} of $\left(\mu_n - \sqrt{\sqgcomplexity{\stcone\cap\twoball{p}}+\sqgcomplexity{\ctcone\cap\twoball{n}}}\right)/\sqrt{n}$, the scaling recommended by \thmref{gen-con-recovery}.}
\label{fig:con-sparse-sparse-noise}
\end{figure*}

Finally, we study the empirical behavior of the constrained recovery program~\eqref{eqn:sig-bound-prob} in the noisy measurement setting, where $\delta \neq 0$. 
We will compare the achieved error in recovering $(x^{\star},v^{\star})$ to the error bounds predicted by \thmref{gen-con-recovery}.
We focus on the recovery of a sparse signal from sparsely corrupted measurements perturbed by dense, unstructured noise.
For fixed noise level $\delta = 1$ and sparsity fractions $(\subsig{\gamma},\subcor{\gamma}) = (0.01, 0.4)$, we vary the signal length $p \in \{100, 300, 500\}$ and the number of measurements $n \in [100,6000]$ and perform the following experiment 20 times for each $(p,n)$ pair:
\begin{enumerate}
\item Generate a signal vector $x^{\star}$ with $p\cdot \subsig{\gamma}$ \iid standard normal entries and $p\cdot (1-\subsig{\gamma})$ entries set to $0$.
\item Draw a Gaussian matrix $\Phi \in \R^{n \times p}$ with independent $N(0,1/n)$ entries.
\item Generate a corruption vector $v^{\star}$ with $n\cdot \subcor{\gamma}$ \iid standard normal entries and $n\cdot (1-\subcor{\gamma})$ entries set to $0$.
\item Solve the following constrained optimization problem with $y = \Phi x^{\star} + v^{\star}$:
\begin{multline*}
(\widehat{x},\widehat{v}) \in \argmin_{x,v}\big\{\norm{v}_1:\\\norm{x}_1\leq \norm{x^{\star}}_1, \norm{y - \Phi x - v}_2 \leq \delta  \big\},
\end{multline*}
\item Record the $\ell_2$ recovery error $\sqrt{\norm{ \widehat{x}-x^{\star}}_2^2 + \norm{ \widehat{v}-v^{\star}}_2^2}$.
\end{enumerate}
Our theory suggests that rescaling the $\ell_2$ recovery error by a factor of 
\begin{align}
\label{eqn:scaling}
n^{-1/2}\left(\mu_n - \sqrt{\sqgcomplexity{\stcone\cap\twoball{p}} + \sqgcomplexity{\ctcone\cap\twoball{n}}}\right),
\end{align}
will give rise to quantity bounded by a universal constant, independent of $n$ or $p$.
This is precisely what we observe in \figref{con-sparse-sparse-noise}.
Here we have plotted, for each setting of $p$ and $n$, the average recovery error and the average rescaled recovery error across 20 iterations.
We have used the optimal squared distance bound \eqref{eqn:SparseDistSquared} to estimate the complexities in the scaling factor \eqref{eqn:scaling}.
We see that while the absolute error curves vary with the choice of $p$, the rescaled error curves converge to a common value, as predicted by the results of \thmref{gen-con-recovery}.


\section{Conclusions and Future Work}
\label{sec:conclusions}
We have presented a geometric approach to analyzing the recovery of structured signals from corrupted and noisy measurements.
We analyzed both penalized and constrained convex programs and, in each case, provided conditions for exact signal recovery from structured corruption and 
stable signal recovery from structured corruption with added unstructured noise.
Our analysis revolved around two geometric measure of signal complexity, the Gaussian complexity and the Gaussian distance, for which we developed new interpretable bounds.
The utility of our theory was borne out in our simulations, which demonstrated close agreement between our theoretical recovery bounds and the sharp phase transitions observed in practice. 
We envision several interesting directions for future work:
\paragraph*{Matching theory in penalized and constrained settings}
In \secref{experiments}, we observed that the penalized convex program \eqref{eqn:pen-prob} with canonical choice of penalty parameter performed nearly as well empirically as the constrained convex program \eqref{eqn:sig-bound-prob} with side information.  
This suggests that our penalized recovery theory could be sharpened to more closely match that obtainable in the constrained recovery setting.

\paragraph*{Penalized noise term}
It would be of great practical interest to analyze the fully penalized convex program
\begin{align*} 
\min_{x,v}\left\{\snorm{x} + \lambda\cnorm{v} + \nu\norm{y-(\Phi x+ v)}_2\right\}
\intertext{or}
\min_{x,v}\left\{\snorm{x} + \lambda\cnorm{v} + \nu\norm{y-(\Phi x+ v)}_2^2\right\}
\end{align*}
when no prior bound $\delta$ on the noise level is available.

\paragraph*{Stochastic noise} Our analysis of stable recovery in the presence of unstructured noise assumes only that the noise is bounded in the $\ell_2$ norm.  
We anticipate improved bounds on estimation error under more specific, stochastic assumptions such as sub-Gaussian or sub-exponential noise.
 
\paragraph*{Non-Gaussian measurements} Finally, an important open question is to what extent the results of this work extend to corrupted sensing problems with non-Gaussian measurements, either stochastic or deterministic, with suitable incoherence conditions.


\clearpage\appendices

\section{Tangent Cone Complexity Bounds}
\label{app:norms}
In this section, we characterize the complexity of the tangent cones generated by the structured vectors and structure-inducing norms introduced in \secref{examples}.
For each tangent cone, we will present new and existing bounds on the squared Gaussian complexity $\sqgcomplexity{\tcone\cap\twoball{p}}$ and the Gaussian squared distance $\sqgdist{t\cdot\partial\norm{x}}$.
The results are summarized in \tabref{WidthBounds} and \tabref{DistBounds}.
Throughout, we use the notation
\[
\dist{x}{C} = \inf_{w\in C} \norm{x- w}_2
\]
to refer to the Euclidean distance between a vector $x\in\R^p$ and a set $C \subset \R^p$.

\subsection{Fundamentals}
Our complexity and distance bounds will be based on two fundamental relationships among Gaussian complexity, Gaussian distance, and a third
measure of size known as the \emph{Gaussian width}:
\[
\oldgwidth{C} \triangleq \EE{\max_{w\in C}\inner{g}{w}}.
\]
The first relation, established in \cite{ChandrasekaranRePaWi12}, upper bounds the Gaussian complexity of a constrained tangent cone in terms of the Gaussian distance and the expected squared distance to the scaled subdifferential 
\begin{multline}\label{eqn:TangentWidthDistance}
 \sqgcomplexity{\tcone\cap\twoball{p}} \leq \min_{t\geq 0} \sqgdist{t\cdot\partial\norm{x}} \\= \min_{t\geq 0}{\EE{\dist{g}{t\cdot\partial\norm{x}}^2}}.
\end{multline}
Our bounds on the right-hand side of \eqref{eqn:TangentWidthDistance} will provide specific settings of the subdifferential scale $t$ that can then be used in choosing a penalty parameter for penalized optimization (see \thmref{gen-pen-recovery}).

In fact, the second fundamental relation provides such a bound on the expected squared distance, in terms of the Gaussian width of the subdifferential, via the following result (proved in \appref{subdiff-proof}):
\begin{proposition}\label{prop:Subdiff}
Let $\tilde{\omega}$ be any lower bound on $\oldgwidth{\partial\norm{x}}$. For $t=\frac{\tilde{\omega}}{\max_{\norm{z}_2=1}\norm{z}}$, 
\[\EE{\dist{g}{t\cdot\partial\norm{x}}^2}\leq p - \left(\frac{\tilde{\omega}}{\max_{\norm{z}_2=1}\norm{z}}\right)^2\;.\]
\end{proposition}
\noindent
We will see that the subdifferentials of many structured norms admit simple lower bounds that lead to tight upper bounds on the Gaussian squared distance and Gaussian squared complexity via \propref{Subdiff}.

\subsection{Sparse vectors} 
We begin by considering an $s$-sparse vector $x\in\R^p$ and the sparsity inducing norm $\norm{\cdot}_1$.
The subdifferential and tangent cone were given in \eqref{eqn:sparse-subdifferential} and \eqref{eqn:ConeExample}
respectively.
In this setting, \citet{ChandrasekaranRePaWi12} established the squared distance bound
\begin{multline*}
\sqgcomplexity{\tcone\cap\twoball{p}} \leq \sqgdist{t\cdot\partial\norm{x}} ={\EE{\dist{g}{t\cdot\partial\norm{x}_1}^2}}\\ \leq{2s\log\left(\frac{p}{s}\right) + \frac{3}{2}s}\;,
\end{multline*}
for $t=\sqrt{2\log(\frac{p}{s})}$.
Using \propref{Subdiff}, we obtain a second bound
\begin{multline}\label{eqn:sparsity_via_subdiff}
\sqgcomplexity{\tcone\cap\twoball{p}} \leq \sqgdist{t'\cdot\partial\norm{x}}\\ ={\EE{\dist{g}{t'\cdot\partial\norm{x}_1}^2}} \leq{p}\left({1 - \nicefrac{2}{\pi}(1-\nicefrac{s}{p})^2}\right),
\end{multline}
for $t'=\sqrt{\nicefrac{2}{\pi}}(1-\nicefrac{s}{p})$. This follows from our more general treatment of block-sparse vectors in \propref{block-distance} below.

\subsection{Block-sparse vectors}
\label{app:block}
Suppose that the indices $\{1,\dots,p\}$ have been partitioned into disjoint blocks $V_1, \dots, V_m$ 
and that $x$ is supported only on $s$ of these blocks.
Then it is natural to consider a norm that encourages block sparsity, e.g., 
\[
\norm{x}_{\ell_1/\ell_2} = \sum_{b=1}^m\norm{x_{V_b}}_2.
\]
\propref{block-distance} presents our two new bounds on Gaussian distance and hence on Gaussian complexity in this setting.

\begin{proposition}[Block-sparse vector Gaussian distance] \label{prop:block-distance}
Partition the indices $\{1,\dots,p\}$ into blocks $V_1,\dots, V_m$ of size $k=p/m$, and let
$\norm{x}_{\ell_1/\ell_2} = \sum_{b=1}^m\norm{x_{V_b}}_2$.
If $x^{\star}$ is supported on at most $s$ of these blocks, then both of the following estimates hold:
\begin{align} \label{eqn:block-distance}
&\textstyle\sqgdist{t\cdot\partial\norm{x}_{\ell_1/\ell_2}} \leq{4s\log(\nicefrac{m}{s}) + (0.5+3k)s} \\&\notag\qtext{ for } t=\sqrt{2\log(m/s)}+\sqrt{k}\;,\\ \label{eqn:block-distance-sub}
&\textstyle\sqgdist{t'\cdot\partial\norm{x}_{\ell_1/\ell_2}} \leq{p}\left({1 - \frac{\mu_k^2}{k}(1-\nicefrac{s}{m})^2}\right) \\\notag&\qtext{ for }t'=\mu_k(1-\nicefrac{s}{m}).
\end{align}
\end{proposition}

For element-wise sparsity ($k=1$), the bound \eqref{eqn:block-distance-sub} specializes to the bound \eqref{eqn:sparsity_via_subdiff} given above. The advantage of exploiting block-sparse structure is evident when one compares the factor of $\frac{\mu_k^2}{k}$ in the bound \eqref{eqn:block-distance-sub} to the term $\frac{2}{\pi}=\mu_1^2$ found in  \eqref{eqn:sparsity_via_subdiff}.
The former is larger and approaches $1$ as the block size $k=p/m$ grows, and hence fewer measurements will suffice to recover from these block-sparse corruptions.

\begin{proof}
Throughout, let $\mathcal{B}$ denote the block indices on which $x$ is supported, with $|\mathcal{B}|=s'\leq s$,
let $\mathcal{T}$ be the space of vectors with support only on those blocks in $\mathcal{B}$,
and let $a \in \mathcal{T}$ be the vector satisfying
\[
a_{V_b} =
\begin{cases}
x_{V_b}/\norm{x_{V_b}}_2 & \text{if } b \in \mathcal{B} \\
0 & \text{otherwise}.
\end{cases}
\]
Then, we have the subdifferential \cite{friedman2010note} 
\[
\partial\norm{x}_{\ell_1/\ell_2} = a + \{w \in \mathcal{T}^\bot :  \max_{b\notin\mathcal{B}}\norm{w_{V_b}}_2 \leq 1\}.
\]

We begin by establishing the bound \eqref{eqn:block-distance-sub} based on the Gaussian width of the subdifferential.
For each $b$,  $\norm{\Pr{V_b}{g}}_2$ is $\chi_k$ distributed.
Using the form of the subdifferential stated above,
\begin{align*}
&\EE{\max_{w\in\partial\norm{x}_{\ell_1/\ell_2}}\inner{g}{w}}\\
	&= \EE{\inner{g}{a} + \sum_{b\notin \mathcal{B}}\sup_{w\in\R^p:\norm{w_{V_b}}_2 \leq 1}\inner{g}{w_{V_b}} } \\
	&= \sum_{b\notin \mathcal{B}}\EE{\sup_{w\in\R^p:\norm{w}_2 \leq 1}\inner{\Pr{V_b}{g}}{w} } \\
	&= \sum_{b\notin \mathcal{B}} \EE{\norm{\Pr{V_b}{g}}_2 } 
	=(m-s')\mu_k \geq (m-s)\mu_k\;.
\end{align*}
Since $\sup_{y\in\R^{p}:\norm{y}_2=1}\norm{y}_{\ell_1/\ell_2}=\sqrt{m}$, \propref{Subdiff} now implies the result  \eqref{eqn:block-distance-sub}.

Now we turn to the bound  \eqref{eqn:block-distance}. By \eqref{eqn:TangentWidthDistance}, it suffices to bound the expected squared distance
\[\EE{\dist{g}{t\cdot\partial\norm{x}_{\ell_1/\ell_2}}^2}\]
 for the given value of $t$.
Note that
\begin{align}
\notag &\EE{\dist{g}{t\cdot\partial\norm{x}_{\ell_1/\ell_2}}^2}\\
\notag	&= \EE{\norm{\Pr{ \mathcal{T} }{g}-ta}_2^2 
	+ \sum_{b\notin\mathcal{B}}\inf_{\norm{z_b}_2\leq t} \norm{\Pr{ V_b }{g}-z_b}_2^2} \\
\notag	&= \EE{\norm{\Pr{ \mathcal{T} }{g}-ta}_2^2 
	+ \sum_{b\notin\mathcal{B}} (\norm{\Pr{ V_b }{g}}_2 - t)_+^2}\\
\notag	&=s'(t^2+p/m) + \sum_{b\notin\mathcal{B}} \EE{(\norm{\Pr{ V_b }{g}}_2 - t)_+^2}\\
\label{eqn:groupsparse-exact}	&=s'(t^2+p/m) + (m-s')\cdot \EE{(\xi-t)_+^2}\\
\notag	&\leq s(t^2+p/m) + (m-s)\cdot \EE{(\xi-t)_+^2}\;,
\end{align}
where $\xi$ is a $\chi_k$ random variable.
We will bound each summand in this expression. 

Letting $d = t-\sqrt{k}$, we have
\begin{align*}
&\EE{(\xi - t)_+^2}\\
	&= \int_0^\infty \PP{(\xi - t)_+^2 \geq a}\ \diff{a}\\
	&= \int_0^\infty \PP{\xi^2  \geq (t+\sqrt{a})^2} \ \diff{a}\\
	&= \int_0^\infty \PP{\xi^2  - k \geq 2\sqrt{k}(d+\sqrt{a}) +(d+\sqrt{a})^2} \ \diff{a}\\
\intertext{Applying the change of variables $c = (d+\sqrt{a})$,}
	&= 2\int_{d}^\infty \PP{\xi^2  - k \geq 2\sqrt{k}c +c^2} (c-d)\ \diff{c} \\
	&\leq 2\int_{d}^\infty e^{-c^2/2} (c-d)\ \diff{c} \\
	&= 2e^{-d^2/2} - 2d\int_{d}^\infty e^{-c^2/2}\ \diff{c} \\
	&\leq  2e^{-d^2/2} - 2\frac{d^2}{d^2+1}e^{-d^2/2}  
	=  \frac{2}{d^2+1}e^{-d^2/2},
\end{align*}
where the penultimate inequality follows from the chi-squared tail bound \citep[Lem.~1]{laurent2000adaptive},
\begin{align*}
\PP{\xi^2 - k \geq \sqrt{2k}c + c^2} \leq e^{-c^2/2} \qtext{for all} c > 0,
\end{align*}
\noindent 
and the final inequality follows from a bound on the Gaussian Q-function
\begin{multline*}
Q(d) 
	= \frac{1}{\sqrt{2\pi}} \int_d^\infty e^{-c^2/2}\ \diff{c}\\
	\geq \frac{1}{\sqrt{2\pi}} \frac{d^2}{d^2+1}\int_d^\infty \frac{c^2+1}{c^2}e^{-c^2/2}\ \diff{c}\\
	= \frac{1}{\sqrt{2\pi}}\frac{d}{d^2+1} e^{-d^2/2} \qtext{for all} d > 0.
\end{multline*}

Hence we have
\begin{align*}
&\EE{\dist{g}{t\cdot\partial\norm{x}_{\ell_1/\ell_2}}^2}\\
	&\leq s(t^2+p/m) + \frac{2(m-s)}{d^2+1}e^{-d^2/2}\\
	&= s(d^2+2d\sqrt{k} + 2k) + \frac{2(m-s)}{d^2+1}e^{-d^2/2} \\
	&\leq s(2d^2 + 3k) + \frac{2(m-s)}{d^2+1}e^{-d^2/2}.
\end{align*}
Since $d = \sqrt{2\log(m/s)}$, we have
\begin{align*}
&\EE{\dist{g}{t\cdot\partial\norm{x}_{\ell_1/\ell_2}}^2}\\
	&\leq s(4\log(m/s)+ 3k) + \frac{2s(1-s/m)}{2\log(m/s)+1} \\
	&\leq (0.5+3k)s + 4s\log(m/s),
\end{align*}
as desired, since 
\[
\frac{2(1-s/m)}{2\log(m/s)+1} < 0.5
\] whenever $1\leq s \leq m$.
\end{proof}


\subsection{Binary vectors} 
The convex hull of the set of binary vectors $\{\pm 1\}^p$ is the unit ball of the $\ell_{\infty}$ norm, so we choose 
$\norm{\cdot}= \norm{\cdot}_{\infty}$ when $x\in\{\pm 1\}^p$.
This choice yields
\[
\partial\norm{x}_{\infty}
	= \{w\in\R^p:\norm{w}_1= 1, w_i x_i\geq 0 \ \forall i\}
\]
and therefore
\[\tcone = \{w\in\R^p:w_ix_i\leq 0 \ \forall i\}\]
and
\[
\ncone = \{w\in\R^p:w_ix_i\geq 0 \ \forall i\}\;.\]
Thus, $\gcomplexity{\tcone\cap\twoball{p}}=\gcomplexity{\ncone\cap\twoball{p}}$, and \cite[Lem.~3.7]{ChandrasekaranRePaWi12} implies that\footnote{
While this lemma is stated in terms of \citet{ChandrasekaranRePaWi12} definitions of width (as discussed in \secref{GaussianWidth}), examining the proof shows that it holds in this case as well.
}
\[\gcomplexity{\tcone\cap\twoball{p}}\leq \sqrt{\frac{p}{2}}\;.\]
We note that a bound based on \propref{Subdiff} is typically looser in this setting.

\subsection{Low-rank matrices} 
When $x\in\R^{m_1\times m_2}$ is a rank $r$ matrix, we consider the trace norm $\norm{\cdot}_*$.
\citet{ChandrasekaranRePaWi12} established the bound $\sqgcomplexity{\tcone\cap \twoball{m_1\times m_2}} \leq{3r(m_1+m_2-r)}$.
\propref{low-rank-distance} presents our new estimate based on the Gaussian width of the subdifferential.

\begin{proposition}[Low-rank matrix Gaussian distance] \label{prop:low-rank-distance}
Let $x\in\R^{m_1\times m_2}$ have rank at most $r$, and let $T$ be the tangent cone associated with $\norm{\cdot}_*$ at $x$.
If $m_1\geq m_2$, then setting  $t = {\frac{4}{27}(m_2-r)\sqrt{m_1-r}/}{m_2}$,
\begin{multline*} 
\textstyle\sqgdist{t\cdot\partial\norm{x}_*}
	\leq\\{m_1m_2}\cdot \left({1 - \left(\frac{4}{27}\right)^2\left(1-\frac{r}{m_1}\right)\left(1-\frac{r}{m_2}\right)^2}\right)\;.
\end{multline*}
\end{proposition}
\begin{proof}
Let $U\Sigma V^\top$ be the compact singular value decomposition of $x$, and 
let $\mathcal{T}$ be the space of matrices in the column or row space of $x$.
The subdifferential is given by \cite{watson1992characterization}
\begin{align*}
\partial\norm{x}_* = UV^\top + \{w \in \mathcal{T}^\bot : \opnorm{w} \leq 1\},
\end{align*}
where $\opnorm{w}$ is the operator norm, i.e.\ the largest singular value, of $w$.

We begin by bounding the Gaussian width of $\partial\norm{x}_*$.
Let $\Gamma\in\R^{m_1\times m_2}$ have \iid standard Gaussian entries.
Since orthogonal projection cannot increase the operator norm, we have
\begin{align*}
&\EE{\max_{w\in\partial\norm{x}_*} \inner{g}{w}}\\
	&= \EE{\inner{\Gamma}{U V^\top} + \sup_{w: \opnorm{\Prp{\mathcal{T}}{w}} \leq 1} \inner{\Gamma}{\Prp{\mathcal{T}}{w}}} \\
	&\geq \EE{\sup_{w: \opnorm{w} \leq 1} \inner{\Prp{\mathcal{T}}{\Gamma}}{w}}
	= \EE{\norm{\Prp{\mathcal{T}}{\Gamma}}_*}.
\end{align*}
Now, suppose that $\tilde{U} = [U\ U']$ and $\tilde{V}=[V\ V']$ are orthonormal for $U'\in\R^{m_1\times (m_1-r)}$
and $V'\in\R^{m_2\times (m_2-r)}$. By definition of $\mathcal{T}$, we know that $\Prp{\mathcal{T}}{\Gamma}=U'U'{}^{\top}\Gamma V'V'{}^{\top}$, and so
\begin{multline*}
\EE{\norm{\Prp{\mathcal{T}}{\Gamma}}_*} 
	= \EE{\norm{U'U'{}^{\top}\Gamma V'V'{}^{\top}}_*} \\
	= \EE{\norm{U'{}^{\top}\Gamma V'}_*} \;,
\end{multline*}
since the trace norm is unitarily invariant. 
Furthermore, since $U'$ and $V'$ each have orthonormal columns, the entries of $U'{}^{\top}\Gamma V'$ are \iid standard Gaussian. We now apply the following lemma (proved in \appref{lemmas}):
\begin{lemma}\label{lem:NuclearNormRandom}
Let $\Gamma\in\R^{m_1\times m_2}$ have \iid standard Gaussian entries, with $m_1\geq m_2$. Then
\[\EE{\norm{\Gamma}_*}\geq \frac{4}{27}m_2\sqrt{m_1}\;.\]
\end{lemma}
Therefore,
\[\EE{\max_{w\in\partial\norm{x}_*} \inner{g}{w}}
 \geq \frac{4}{27}(m_2-r)\sqrt{m_1-r}\;.\]
 Finally, we apply \propref{Subdiff} to obtain the desired bound.
\end{proof}

\subsection{Proof of subdifferential distance bound}\label{app:subdiff-proof}

We begin with a definition and a lemma (proved in \appref{lemmas}). 
The {\em dual norm} to $\norm{\cdot}$ is defined as
\[\norm{w}^*=\sup_{\norm{x}\leq 1}\inner{w}{x}\;,\]
and satisfies
\[\inner{w}{x}\leq \norm{x}\cdot\norm{w}^*\text{ for all }w,x\in\R^p\;.\]
\begin{lemma}\label{lem:NormRatios}
For any norm $\norm{\cdot}$ on $\R^p$ with dual norm $\norm{\cdot}^*$,
\[\max_{x\in\R^p}\frac{\norm{x}}{\norm{x}_2} = \max_{x\in\R^p}\frac{\norm{x}_2}{\norm{x}^*}\;.\]
\end{lemma}

Now we prove our Gaussian distance bound that is based on the subdifferential. We use the fact that $\norm{w}^*=1$ for all $w\in\partial\norm{x}$.
\begin{repproposition}{prop:Subdiff}
Let $\tilde{\omega}$ be any lower bound on $\EE{\max_{w\in\partial\norm{x}}\inner{g}{w}}$. For $t=\frac{\tilde{\omega}}{\max_{\norm{z}_2=1}\norm{z}}$, 
\[\EE{\dist{g}{t\cdot\partial\norm{x}}^2}\leq p - \left(\frac{\tilde{\omega}}{\max_{\norm{z}_2=1}\norm{z}}\right)^2\;.\]
\end{repproposition}
\begin{proof}
Fix any $g\in\R^p$, and choose any
\[w_0\in\arg\max_{w\in\partial\norm{x}}\inner{g}{w}\;.\]
(Since $\partial\norm{x}$ is closed and is in the unit sphere of the norm $\norm{\cdot}^*$, it is compact, and so the maximum must be attained at some $w_0$.) Then for any $t$,
\begin{align*}&\dist{g}{t\cdot\partial\norm{x}}^2\\
&\leq \norm{g-t\cdot w_0}^2_2=\norm{g}^2_2-2t\inner{g}{w_0}+t^2\norm{w_0}^2_2\\
&=\norm{g}^2_2-2t\max_{w\in\partial\norm{x}}\inner{g}{w}+t^2\norm{w_0}^2_2\\
\intertext{Since $\norm{w}^*=1$ for all $w\in\partial\norm{x}$,}
&\leq \norm{g}^2_2-2t\max_{w\in\partial\norm{x}}\inner{g}{w}+t^2\max_{\norm{w}^*=1}\norm{w}^2_2\\
&\leq \norm{g}^2_2-2t\max_{w\in\partial\norm{x}}\inner{g}{w}+t^2\max_{\norm{w}_2=1}\norm{w}^2\;,
\end{align*}
where the last step comes from \lemref{NormRatios}. Taking expectations,
\begin{multline*}\EE{\dist{g}{t\cdot\partial\norm{x}}^2}\leq\\ p -2t\cdot \EE{\max_{w\in\partial\norm{x}}\inner{g}{w}}+t^2\cdot \max_{\norm{w}_2=1}\norm{w}^2\\\leq p -2t\cdot \tilde{\omega}+t^2\cdot \max_{\norm{w}_2=1}\norm{w}^2\;.\end{multline*}
Now we can minimize this quadratic in $t$ by setting $t=\frac{\tilde{\omega}}{\max_{\norm{w}_2=1}\norm{w}^2}$, which yields the desired bound.
\end{proof}

\subsection{Relating Gaussian distance and Gaussian complexity}\label{app:dist-vs-complexity-proof}
\begin{repproposition}{prop:dist-vs-complexity}
Suppose that, for $x\neq 0$, $\partial\norm{x}$ satisfies a \emph{weak decomposability} assumption:
\begin{equation}\label{eqn:subdiff-decomp}
\exists w_0\in\partial\norm{x} \text{ s.t. }\inner{w-w_0}{w_0}=0 \ \forall w\in\partial\norm{x}\;.\end{equation}
Then
\[\min_{t\geq0}\gdist{t\cdot\partial\norm{x}}\leq\gcomplexity{\tcone\cap\twoball{p}}+6\;.\]
\end{repproposition}

\begin{proof}
For $g\in\R^p$, define\footnote{
Note that $ t _g$ is unique because $\ncone=\cup_{ t \geq 0} t \cdot\partial\norm{x}$ is convex, and so there is a unique projection of $g$ onto this cone; since $\norm{w}^*=1$ for all $w\in  t \cdot\partial\norm{x}$, there is no overlap between $ t \cdot\partial\norm{x}$ and $ t '\cdot\partial\norm{x}$ for any $ t \neq  t '$, and so $ t _g$ is defined uniquely from the projection of $g$ onto $\ncone$.
}
\begin{equation}\label{eqn:best-t-for-g} t _g = \arg\min_{ t \geq 0}\dist{g}{ t \cdot\partial\norm{x}}\;.\end{equation}

Define the event
\[\mathcal{E} = \left\{\left| t _g-\EE{ t _g}\right|< \frac{2}{\norm{w_0}_2}\right\}\;.\]
We will use the following lemma (proved in \appref{lemmas}):
\begin{lemma}\label{lem:t-hat-Lipschitz}
Suppose that, for $x\neq 0$, $\partial\norm{x}$ satisfies \eqref{eqn:subdiff-decomp}. Let $t_g$ be defined as in \eqref{eqn:best-t-for-g}. 
Then $g\mapsto  t _g$ is a $\frac{1}{\norm{w_0}_2}$-Lipschitz function of $g$.
\end{lemma}
Therefore, applying a bound due to \citet[(2.8)]{Ledoux},
\[\PP{\mathcal{E}}\geq 1 - 2e^{-\nicefrac{2^2}{2}}\;.\]

Now suppose that $\mathcal{E}$ holds.
 Find $w\in\partial\norm{x}$ such that $ t _g\cdot w$ is the projection of $g$ to $\ncone$, that is,
\[\dist{g}{\ncone}=\dist{g}{ t _g\cdot\partial\norm{x}}=\norm{g- t _g\cdot w}_2\;.\]
Since the subdifferential $\partial\norm{x}$ is convex, we have
\[\frac{ t _g}{\EE{ t _g}+\frac{2}{\norm{w_0}_2}}\cdot w + \frac{\EE{t_g}+\frac{2}{\norm{w_0}_2}-t_g}{\EE{ t _g}+\frac{2}{\norm{w_0}_2}}\cdot w_0 \in \partial\norm{x}\;,\]
and so
\begin{align*}
&\dist{g}{(\EE{ t _g}+\frac{2}{\norm{w_0}_2})\cdot\partial\norm{x}}\\
&\leq \norm{g-\left[ t _g\cdot w + (\EE{t_g}+\frac{2}{\norm{w_0}_2}-t_g)\cdot w_0\right]}_2\\
&\leq \norm{g- t _g\cdot w}_2 + (\EE{t_g}+\frac{2}{\norm{w_0}_2}-t_g)\cdot \norm{w_0}_2\\
&=\dist{g}{\ncone}+(\EE{t_g}+\frac{2}{\norm{w_0}_2}-t_g)\cdot\norm{w_0}_2\\
&<\dist{g}{\ncone}+4\;,\end{align*}
where the last step comes from the definition of the event $\mathcal{E}$.
Therefore,
\begin{multline*}\mathbb{P}\bigg\{\frac{1}{2}\dist{g}{(\EE{ t _g}+\frac{2}{\norm{w_0}_2})\cdot\partial\norm{x}} -\\ \frac{1}{2}\dist{g}{\ncone}< 2\bigg\} \\\geq \PP{\mathcal{E}}\geq 1 - 2 e^{-\nicefrac{2^2}{2}}\;.\end{multline*}
Since $g\mapsto\dist{g}{A}$ is clearly a $1$-Lipschitz function of $g$ for any convex set $A$, we see that
\[g\mapsto \frac{1}{2}\dist{g}{(\EE{ t _g}+\frac{2}{\norm{w_0}_2})\cdot\partial\norm{x}} - \frac{1}{2}\dist{g}{\ncone}\]
is a $1$-Lipschitz function of $g$. We will now make use of the following lemma (proved in \appref{lemmas}):
\begin{lemma}\label{lem:Lipschitz-E-vs-P}
Let $\phi:\R^p\rightarrow \R$ be a $1$-Lipschitz function. For any $a\in\R$ and $p_0>0$, 
\[\PP{\phi(g)< a}\geq p_0 \quad \Rightarrow \quad \EE{\phi(g)}\leq a + \sqrt{2\log(\nicefrac{1}{p_0})}\;.\]
\end{lemma}
Therefore,
\begin{align*}
&\frac{\gdist{(\EE{ t _g}+\frac{2}{\norm{w_0}_2})\cdot\partial\norm{x}}-\gdist{\ncone}}{2}\\
&=\EE{\frac{1}{2}\dist{g}{(\EE{ t _g}+\frac{2}{\norm{w_0}_2})\cdot\partial\norm{x}}}\\
&\quad-\EE{ \frac{1}{2}\dist{g}{\ncone}}\\
&\leq 2+\sqrt{2\log(\frac{1}{1 - 2e^{-\nicefrac{2^2}{2}}})}\;.\end{align*}
Since $\gdist{\ncone}=\gcomplexity{\tcone\cap\twoball{p}}$, we see that
 \begin{align*}
 &\gdist{(\EE{ t _g}+\frac{2}{\norm{w_0}_2})\cdot\partial\norm{x}}\\
 &\leq \gcomplexity{\tcone\cap\twoball{p}}+ 2\cdot\left(2+\sqrt{2\log(\frac{1}{1 - 2e^{-\nicefrac{2^2}{2}}})}\right)\\
 & \leq \gcomplexity{\tcone\cap\twoball{p}} + 6\;.\end{align*}
\end{proof}

\section{Proofs of Application Results}
\label{app:proofs-applications}

\subsection{Corollaries}

First we prove our result on constrained recovery of a binary signal corrupted with sparse noise.
\begin{proof}[Proof of \corref{BinaryPlusSparse}]
The bound \eqref{eqn:SparseConeBound} follows from the $\ell_1$ tangent cone complexity bounds of \tabref{WidthBounds}.
Next, applying \thmref{gen-con-recovery} with any $\eps\in (2\delta,\frac{\mu_n-\tau}{\sqrt{n}})$ (and making use of the $\ell_\infty$ tangent cone complexity bound of \tabref{WidthBounds}), we see that with probability at least $1-e^{-(\mu_n-\eps\sqrt{n}-\tau)^2/2}$,
\[\norm{\widehat{x}-x^{\star}}_\infty\leq \sqrt{\norm{\widehat{x}-x^{\star}}^2_2+\norm{\widehat{v}-v^{\star}}^2_2}\leq \frac{2 \delta}{\eps}<1\;.\]
This implies that $x^{\star}$ is the nearest binary vector to $\widehat{x}$, that is, we can exactly recover $x^{\star}$ by setting
$x^{\star}=\sign(\widehat{x})$. Letting $\eps$ approach $2\delta$, we see that this is true with probability at least $1-e^{-(\mu_n-2\delta\sqrt{n}-\tau)^2/2}$.
\end{proof}

Next we prove our result on penalized recovery of a structured signal observed with a high frequency of block-wise corruptions.
\begin{proof}[Proof of \corref{pen-gen-block}]
\thmref{gen-pen-recovery} implies stable recovery with error at most $2\delta/\eps$ and probability at least $1-\beta$ once
\begin{multline*}
\mu_n - \sqrt{n}\eps \geq \\2\gdist{\subsig{t}\cdot\partial\snorm{x^{\star}}}+\gdist{\mu_k(1-\gamma)\cdot\partial\norm{v^{\star}}_{\ell_1/\ell_2}} \\+  \sqrt{2\log(\nicefrac{1}{\beta})}+ 3\sqrt{2\pi}+ \frac{1}{\sqrt{2}}+\frac{1}{\sqrt{2\pi}}\;.
\end{multline*}
Since $\mu_n > \sqrt{n-1/2} \geq \sqrt{n} - (1-1/\sqrt{2}) $ for all $n\geq 1$, it is sufficient to have
\begin{multline*}
\sqrt{n}(1-\eps) \geq \\2\gdist{\subsig{t}\cdot\partial\snorm{x^{\star}}}+\gdist{\mu_k(1-\gamma)\cdot\partial\norm{v^{\star}}_{\ell_1/\ell_2}} \\+  \sqrt{2\log(\nicefrac{1}{\beta})}+ 3\sqrt{2\pi}+ 1+\frac{1}{\sqrt{2\pi}}\;.
\end{multline*}
\tabref{DistBounds} bounds $\gdist{\mu_k(1-\gamma)\cdot\partial\norm{v^{\star}}_{\ell_1/\ell_2}}$ by $\sqrt{n}\sqrt{1 - \nicefrac{\mu_k^2}{k}(1-\gamma)^2}$, and hence the condition that $\sqrt{n}$ is at least as large as
\[
 \frac{2\gdist{\subsig{t}\cdot\partial\snorm{x^{\star}}}+  \sqrt{2\log(\nicefrac{1}{\beta})}+ 3\sqrt{2\pi}+ 1+\frac{1}{\sqrt{2\pi}}}{
\left(1-\sqrt{1 - \nicefrac{\mu_k^2}{k}(1-\gamma)^2} -\eps\right)_+}\;
\]
is sufficient for recovery.
To obtain the final form of our statement, we note that
\begin{multline*}
1 - \sqrt{1-\nicefrac{\mu_k^2}{k}(1-\gamma)^2} = 
\frac{\nicefrac{\mu_k^2}{k}(1-\gamma)^2}{1+\sqrt{1-\nicefrac{\mu_k^2}{k}(1-\gamma)^2}}
\\\geq 
\frac{\nicefrac{\mu_k^2}{k}}{1+\sqrt{1-\nicefrac{\mu_k^2}{k}}}(1-\gamma)^2
=
\left(1 - \sqrt{1-\nicefrac{\mu_k^2}{k}}\right)(1-\gamma)^2.
\end{multline*}
\end{proof}

Finally we prove our result on penalized recovery of an extremely sparse signal observed under extremely sparse corruptions.
 \begin{proof}[Proof of \corref{pen-extreme-sparsity}]
 To establish the result with $\lambda=1$, we choose
 \[\subsig{t}=\subcor{t}=\sqrt{2\log\left(\frac{p+n}{\subsig{s}+\subcor{s}}\right)}\;.\]
 To compute the relevant expected squared distances, we use the fact that, by \eqref{eqn:SparseDistSquared},
 \begin{align*}&\EE{\dist{g}{\subsig{t}\cdot\partial\norm{x^{\star}}_1}^2}\\\
 &\leq \subsig{s}(1+\subsig{t}^2)+\frac{2(p-\subsig{s})}{\sqrt{2\pi}}\cdot \\
 &\quad \left((1+\subsig{t}^2)\int_{\subsig{t}}^{\infty}\exp(-c^2/2)\;\ \diff{c} - \subsig{t}\exp(-\subsig{t}^2/2)\right)\\
 \intertext{Bounding the integral as in \cite[Proof of Prop. 3.10]{ChandrasekaranRePaWi12}, }
 &\leq  \subsig{s}(1+\subsig{t}^2)+\frac{2(p-\subsig{s})}{\sqrt{2\pi}}\cdot \frac{1}{\subsig{t}}\exp(-\subsig{t}^2/2)\\
 &\leq  \subsig{s}\left(1+2\log\left(\frac{p+n}{\subsig{s}+\subcor{s}}\right)\right)\\
 &\quad  +\frac{2(p-\subsig{s}+n-\subcor{s})}{\sqrt{2\pi}}\cdot \frac{1}{\sqrt{2\log\left(\frac{p+n}{\subsig{s}+\subcor{s}}\right)}}\\
 &\quad\quad\cdot \frac{\subsig{s}+\subcor{s}}{p+n}\\
 &\leq \subsig{s}\left(1+2\log\left(\frac{p+n}{\subsig{s}+\subcor{s}}\right)\right)+\frac{1}{2}\left(\subsig{s}+\subcor{s}\right)\\
 &\leq (\subsig{s}+\subcor{s})\cdot\left(2\log\left(\frac{p+n}{\subsig{s}+\subcor{s}}\right)+\frac{3}{2}\right)\;,
 \end{align*}
 since
 \[
 \frac{1-(\subsig{s}+\subcor{s})/(p+n)}{\sqrt{\pi\log((p+n)/(\subsig{s}+\subcor{s}))}} \leq 1/2
 \]
 whenever $1 \leq \subsig{s}+\subcor{s} \leq p+n$.
 We also have the analogous bound for the expected Gaussian distance for the corruption, and combining the two, we get
 \begin{multline*}2\gdist{\subsig{t}\cdot\partial\norm{x^{\star}}_1}+\gdist{\subcor{t}\cdot\partial\norm{v^{\star}}_1} \\\leq 
 3\sqrt{(\subsig{s}+\subcor{s})\cdot\left(2\log\left(\frac{p+n}{\subsig{s}+\subcor{s}}\right)+\frac{3}{2}\right)}\;.\end{multline*}
 The result then follows by applying \thmref{gen-pen-recovery}.
 \end{proof}

\subsection{Scaling for sparse signal and dense corruption}\label{app:LiCompare}
\begin{corollary}\label{cor:LiCompare}
Fix any corruption proportion $\gamma<1$ and let $s_\gamma$ be defined as in \eqref{eqn:LiCompare}.
If $x^{\star}$ and $v^{\star}$ have at most $s_\gamma$ and $ n\cdot\gamma$ nonzero entries, respectively, and $(\widehat{x},\widehat{v})$ is a solution to 
\[\arg\min_{(x,v)}\left\{\norm{x}_1+\lambda\norm{v}_1:\norm{y-(\Phi x+v)}_2\leq \delta\right\}\]
with
\[\lambda = \frac{1-\gamma}{\sqrt{\pi\log(\nicefrac{p}{s_\gamma})}}\;,\]
then with probability at least $0.9$, the recovery error satisfies
\[\sqrt{\norm{\widehat{x}-x^{\star}}^2_2+\norm{\widehat{v}-v^{\star}}^2_2}\leq\frac{2\delta}{\eps}\;.\]
\end{corollary}
\begin{proof}
If $s_\gamma = 0$, then $\lambda = 0$, $\widehat{x} = x^{\star} = 0$, and $\norm{y-(\Phi \widehat{x}+\widehat{v})}_2 = \norm{v^{\star} + z - \widehat{v}}_2 \leq \delta$.
Hence, with probability 1, $\sqrt{\norm{\widehat{x}-x^{\star}}^2_2+\norm{\widehat{v}-v^{\star}}^2_2}\leq 2\delta\leq 2\delta/\eps$.

Otherwise, $s_\gamma \geq 1$.
With the choice $\subsig{t}=\sqrt{2\log(p/s_\gamma)}$, \corref{pen-gen-block} and the prior sparse distance bound of \tabref{DistBounds} together imply that having $\sqrt{n}$ at least as large as
\[
\frac{2\sqrt{2s_\gamma\log(\nicefrac{p}{s_\gamma})+\frac{3}{2}s_\gamma} + \sqrt{2\log(\nicefrac{1}{0.1})} +3\sqrt{2\pi}+ 1 +\frac{1}{\sqrt{2\pi}}}{\left(\alpha_1\left(1-\gamma\right)^2 - \eps\right)_+}
\]
suffices to achieve the desired recovery guarantee.
Since $s_\gamma \geq 1$, we have
\begin{multline*}
\sqrt{2\log(10)} +3\sqrt{2\pi}+ 1 +\frac{1}{\sqrt{2\pi}} \leq 11.1 \leq\\ 10\sqrt{2s_\gamma\log(\nicefrac{p}{s_\gamma})+\frac{3}{2}s_\gamma}
\end{multline*}
and hence
\[
\sqrt{n} 
	\geq \frac{12\sqrt{2s_\gamma\log(\nicefrac{p}{s_\gamma})+\frac{3}{2}s_\gamma}}{\left(\alpha_1\left(1-\gamma\right)^2 - \eps\right)_+}
\]
is sufficient for recovery.
One can check that the chosen $s_\gamma$ always satisfies this bound.
\end{proof}

\section{Proofs of Main Results}\label{app:main-proofs}
\subsection{Main idea}
Under the assumptions of \secref{problem}, we have an $\ell_2$ bound on the noise level of our problem,
\[\norm{y- (\Phi x^{\star}+v^{\star})}_2\leq \delta\;,\]
and hence our feasible set is
\[\left\{(x,v)\in\R^{p\times n}:\norm{y-(\Phi x+v)}_2\leq \delta\right\}\;.\]

\paragraph*{Tangent cones} We will refer to the signal and corruption tangent cones
\[\subsig{T}=\left\{a\in\R^p:\exists t>0,\snorm{x^{\star}+t\cdot a}\leq \snorm{x^{\star}}\right\}\]
and
\[\subcor{T}=\left\{b\in\R^n:\exists t>0,\cnorm{v^{\star}+t\cdot b}\leq \cnorm{v^{\star}}\right\}\;.\]
Given a penalty parameter $\lambda\in(0,\infty)$, we write a joint penalty function
\[\mathrm{Pen}_\lambda(x,v)=\snorm{x}+\lambda\cnorm{v}\]
and define the {\em joint tangent cone} given by
\begin{multline*}\jtcone=\big\{(a,b)\in\R^p\times\R^n:\\\mathrm{Pen}_\lambda(x^{\star}+t\cdot a,v^{\star}+t\cdot b)\leq \mathrm{Pen}_{\lambda}(x^{\star},v^{\star})\text{ for some }t>0\big\}\;.\end{multline*}
By definition,
\[\subsig{T}\times \subcor{T} \subsetneq \jtcone\subsetneq (\subsig{T}\times \R^n)\cup(\R^p\times \subcor{T})\;.\]

\paragraph*{Constrained problem} Consider either version of the constrained estimation problem:
\begin{multline*}(\widehat{x},\widehat{v})=\arg\min\big\{\snorm{x}:\\\cnorm{v}\leq\cnorm{v^{\star}},\norm{y-(\Phi x+ v)}_2\leq \delta\big\}\end{multline*}
or
\begin{multline*}(\widehat{x},\widehat{v})=\arg\min\big\{\cnorm{v}:\\\snorm{x}\leq\snorm{x^{\star}},\norm{y-(\Phi x+ v)}_2\leq \delta\big\}\end{multline*}
In both optimization problems,  our solution $(\widehat{x},\widehat{v})$ will necessarily satisfy
\[\snorm{\widehat{x}}\leq \snorm{x^{\star}}\text{ and }\cnorm{\widehat{v}}\leq\cnorm{v^{\star}}\;.\]
We know therefore that
\[(\widehat{x}-x^{\star},\widehat{v}-v^{\star})\in \subsig{T}\times \subcor{T}\;.\]
 Below, we will derive a high probability lower bound on
 \[\min_{(a,b)\in (\subsig{T}\times \subcor{T})\backslash\{(0,0)\}}\frac{\norm{\Phi a + b}_2}{\sqrt{\norm{a}^2_2+\norm{b}^2_2}}\;,\]
or, equivalently after rescaling,
\[\min_{(a,b)\in (\subsig{T}\times \subcor{T})\cap\sphere{p+n}}\norm{\Phi a + b}_2\;.\]

Given this bound, since we know that
\begin{multline*}\norm{\Phi(\widehat{x}-x^{\star})+(\widehat{v}-v^{\star})}_2\leq \\\norm{y-(\Phi\widehat{x}+\widehat{v})}_2+\norm{y-(\Phi x^{\star}+v^{\star})}_2\leq 2\delta\;,\end{multline*}
we obtain the estimation error bound
\begin{multline*}\sqrt{\norm{\widehat{x}-x^{\star}}^2_2+\norm{\widehat{v}-v^{\star}}^2_2}\leq \\2\delta\cdot \left(\min_{(a,b)\in (\subsig{T}\times \subcor{T})\cap\sphere{p+n}}\norm{\Phi a + b}_2\right)^{-1}\;.\end{multline*}

\paragraph*{Penalized problem} Consider the penalized optimization problem
\[(\widehat{x},\widehat{v})=\arg\min\left\{\mathrm{Pen}_\lambda(x,v):\norm{y-(\Phi x+ v)}_2\leq \delta\right\}\;.\]
Our solution $(\widehat{x},\widehat{v})$ will satisfy the weaker condition
\[\mathrm{Pen}_\lambda(\widehat{x},\widehat{v})\leq \mathrm{Pen}_\lambda(x^{\star},v^{\star})\;,\]
and so we will have the weaker inclusion
\[(\widehat{x}-x^{\star},\widehat{v}-v^{\star})\in \jtcone\;.\]
Following the same reasoning as in the constrained case,
we obtain the estimation error bound
\begin{multline*}\sqrt{\norm{\widehat{x}-x^{\star}}^2_2+\norm{\widehat{v}-v^{\star}}^2_2}\leq\\ 2\delta\cdot \left(\min_{(a,b)\in \jtcone\cap\sphere{p+n}}\norm{\Phi a + b}_2\right)^{-1}\;.\end{multline*}

\subsection{Preliminaries}
In both the constrained and the penalized setting, we begin by relating $\min_{(a,b)\in\Omega}\norm{\Phi a + b}_2$ to a second Gaussian functional via the following lemma (proved in \appref{lemmas}):
\begin{lemma}\label{lem:GaussianInequality}
Take any 
$\Omega\subset\twoball{p}\times \R^n$.
Let $\Phi\in\R^{n\times p}$ have \iid $N(0,\frac{1}{n})$ entries, and let $g\in\R^p$ and $h\in\R^n$ have \iid $N(0,1)$ entries.
Then
\begin{multline}\label{eqn:Phi_g_h}\sqrt{n}\cdot \EE{\min_{(a,b)\in\Omega}\norm{\Phi a + b}_2} \geq\\ \EE{\left(\min_{(a,b)\in\Omega}\norm{h\cdot\norm{a}_2+\sqrt{n}\cdot b}_2 + \inner{g}{a}\right)_+} -\frac{1}{\sqrt{2\pi}}\;.\end{multline}
\end{lemma}

Next, consider the matrix $\sqrt{n}\cdot \Phi$, which has \iid standard Gaussian entries. We know that $\sqrt{n}\cdot \min_{(a,b)\in\Omega}\norm{\Phi a + b}_2$ is a $1$-Lipschitz function of the matrix $\sqrt{n}\cdot \Phi$, because of the assumption that $\norm{a}_2\leq 1$ for all $a\in\Omega$. Therefore, applying a bound of \citet[(2.8)]{Ledoux}, we see that with probability at least $1-\beta$,
\begin{multline}\label{eqn:Phi_EPhi}\sqrt{n}\cdot \min_{(a,b)\in\Omega}\norm{\Phi a + b}_2 \geq \\\sqrt{n}\cdot \EE{\min_{(a,b)\in\Omega}\norm{\Phi a + b}_2}  - \sqrt{2\log(\nicefrac{1}{\beta})}\;,\end{multline}
for any $\beta\in(0,1]$. 

These two bounds form the main ingredients of our proofs. In the next two sections, we will lower-bound the expectation in the right-hand side of \eqref{eqn:Phi_g_h} for $\Omega=(\subsig{T}\times \subcor{T})\cap\sphere{p+n}$ and for $\Omega= \jtcone\cap\sphere{p+n}$, to handle the constrained and penalized problems, respectively. Combining these results with the probability bound in \eqref{eqn:Phi_EPhi}, our main results are obtained.

\subsection{Lower bound: constrained setting}
In this section we derive a lower bound on
\[\min_{\substack{(a,b)\in\\(\subsig{T}\times \subcor{T})\cap\sphere{p+n}}}\norm{\Phi a + b}_2\;.\]
Let
\[\omega^2=\sqgcomplexity{\subsig{T}\cap\twoball{p}}\quad\text{and}\quad\omega'{}^2=\sqgcomplexity{\subcor{T}\cap \twoball{n}}\]
be shorthand for the two relevant Gaussian squared complexities.
\begin{theorem}\label{thm:LowerBoundConstrained}
\begin{multline*}\sqrt{n}\EE{\min_{\substack{(a,b)\in\\(\subsig{T}\times \subcor{T})\cap\sphere{p+n}}}\norm{\Phi a + b}_2} \geq\\ \mu_n -\sqrt{\omega^2+\omega'{}^2}-\frac{1}{\sqrt{2}} -\frac{1}{\sqrt{2\pi}} \;.\end{multline*}
\end{theorem}

\begin{proof}
We will combine \lemref{GaussianInequality} with a lower bound for
\[\EE{\left(\min_{\substack{(a,b)\in\\(\subsig{T}\times \subcor{T})\cap\sphere{p+n}}} \!\!\!\!\!\!\!\!\norm{h\cdot \norm{a}_2+\sqrt{n}\cdot b}_2+\inner{g}{a}\right)_+}.\]

Let
\[\widehat{\omega}=\max_{a\in \subsig{T}\cap \twoball{p}}\inner{-g}{a} \quad\text{and}\quad\widehat{\omega}'=\max_{b\in \subcor{T}\cap\twoball{n}}\inner{-h}{b}\]
be the `observed' Gaussian complexities, with expectations $\omega$ and $\omega'$, where $g\in\R^p$ and $h\in\R^n$ are independent vectors with \iid standard normal entries.

If $\sqrt{\widehat{\omega}^2+\widehat{\omega}'{}^2}>\norm{h}_2$, then 
\begin{multline*}\left(\min_{\substack{(a,b)\in\\(\subsig{T}\times \subcor{T})\cap\sphere{p+n}}}\norm{h\cdot\norm{a}_2+\sqrt{n}\cdot b}_2+\inner{g}{a}\right)_+\\\geq 0 \geq \norm{h}_2 - \sqrt{\widehat{\omega}^2+\widehat{\omega}'{}^2}\;.\end{multline*}
If not, then by definition of $(\subsig{T}\times \subcor{T})\cap\sphere{p+n}$,
\begin{align*}
&\min_{\substack{(a,b)\in\\(\subsig{T}\times \subcor{T})\cap\sphere{p+n}}}\norm{h\cdot\norm{a}_2+\sqrt{n}\cdot b}_2+\inner{g}{a}\\
&=\min_{\substack{c\in[0,1]\\a\in \subsig{T}\cap\sphere{p}\\b\in \subcor{T}\cap\sphere{n}}}\norm{h\cdot c+b\cdot \sqrt{n}\cdot \sqrt{1-c^2}}_2+c\cdot \inner{g}{a}\\
&\geq \min_{\substack{c\in[0,1]\\a\in \subsig{T}\cap\sphere{p}\\b\in \subcor{T}\cap\sphere{n}}}\bigg\{\norm{h\cdot c+b\cdot \norm{h}_2\sqrt{1-c^2}}_2\\&\quad \quad \quad +c\cdot \inner{g}{a} - \left|\sqrt{n} - \norm{h}_2\right|\bigg\}\\
&= \min_{\substack{c\in[0,1]\\a\in \subsig{T}\cap\sphere{p}\\b\in \subcor{T}\cap\sphere{n}}}\bigg\{\sqrt{\norm{h}^2_2 + 2\norm{h}_2\cdot \inner{h}{b}\cdot c\sqrt{1-c^2}}\\&\quad \quad \quad +c\cdot \inner{g}{a} - \left|\sqrt{n} - \norm{h}_2\right|\bigg\}\\
&\geq \min_{\substack{c\in[0,1]\\a\in \subsig{T}\cap\twoball{p}\\b\in \subcor{T}\cap\twoball{n}}}\bigg\{\sqrt{\norm{h}^2_2+ 2\norm{h}_2\cdot \inner{h}{b}\cdot c\sqrt{1-c^2}}\\& \quad \quad \quad +c\cdot \inner{g}{a} - \left|\sqrt{n} - \norm{h}_2\right|\bigg\}\\
\intertext{Minimizing over $a$ and $b$ for a fixed $c$,}
&= \min_{c\in[0,1]}\bigg\{\sqrt{\norm{h}^2_2 - 2\norm{h}_2\cdot \widehat{\omega}'\cdot c\sqrt{1-c^2}}\\
&\quad- c\cdot \widehat{\omega} - \left|\sqrt{n} - \norm{h}_2\right|\bigg\}\\
&\geq \norm{h}_2 - \sqrt{\widehat{\omega}^2+\widehat{\omega}'{}^2} - \left|\sqrt{n} - \norm{h}_2\right|\;,
\end{align*}
where for the last step, we use the fact that $\sqrt{\widehat{\omega}^2+\widehat{\omega}'{}^2}\leq\norm{h}_2$, and apply the following lemma (proved in \appref{lemmas}):
\begin{lemma}\label{lem:abc_Identity}
For any $a,b,C\geq 0$ such that $a^2+b^2\leq C^2$,
\[\inf_{u\in[0,1]} \sqrt{C^2-2Cu\sqrt{1-u^2}\cdot a}-u\cdot b\geq C - \sqrt{a^2+b^2}\;.\]
\end{lemma}

In either case, then, we have proved that
\begin{multline*}\left(\min_{\substack{(a,b)\in\\(\subsig{T}\times \subcor{T})\cap\sphere{p+n}}}\norm{h\cdot\norm{a}_2+\sqrt{n}\cdot b}_2+\inner{g}{a}\right)_+\\\geq \norm{h}_2 - \sqrt{\widehat{\omega}^2+\widehat{\omega}'{}^2}- \left|\sqrt{n} - \norm{h}_2\right|\;.\end{multline*}
Finally, we take expectations. To bound the first subtracted term, we have 
\[\EE{ \sqrt{\widehat{\omega}^2+\widehat{\omega}'{}^2}}\leq \sqrt{\EE{\widehat{\omega}^2+\widehat{\omega}'{}^2}}=\sqrt{\omega^2+\omega'{}^2}\;,\]
by definition of the Gaussian squared complexity. To bound the second subtracted term, we use the following lemma (proved in \appref{lemmas}):
\begin{lemma}\label{lem:GaussNorm}
For $h\sim N(0,I_n)$,
$\EE{\left|\norm{h}_2-\sqrt{n}\right|}\leq \frac{1}{\sqrt{2}}$.
\end{lemma}

Finally, applying \lemref{GaussianInequality}, this gives us the desired lower bound.
\end{proof}

\subsection{Lower bound: penalized setting}
In this section we derive a lower bound on
\[\min_{(a,b)\in\jtcone\cap\sphere{p+n}}\norm{\Phi a + b}_2\;.\]

\begin{theorem}\label{thm:pen-bound}
Let $\lambda={\subcor{t}}{/\subsig{t}}$ for parameters $\subsig{t},\subcor{t}\geq 0$.  Then
\begin{multline*}\sqrt{n}\EE{\min_{(a,b)\in\jtcone\cap\sphere{p+n}}\norm{\Phi a + b}_2}\geq \\\mu_n -  2\gdist{\subsig{t}\cdot\partial\snorm{x^{\star}}}-\gdist{\subcor{t}\cdot\partial\cnorm{v^{\star}}} \\- 3\sqrt{2\pi}- \frac{1}{\sqrt{2}}-\frac{1}{\sqrt{2\pi}}\;.\end{multline*}
\end{theorem}
\begin{proof}

We will combine \lemref{GaussianInequality} with a lower bound for
\[\EE{\left(\min_{(a,b)\in\jtcone\cap\sphere{p+n}} \norm{h\cdot \norm{a}_2+\sqrt{n}\cdot b}_2+\inner{g}{a}\right)_+}\;.\]

For $\sigma=+1$ and $\sigma=-1$, define
\[\subsig{d}^{(\sigma)}\coloneqq \dist{\sigma\cdot g}{\subsig{t}\cdot\partial\snorm{x^{\star}}}\]
and
\[\subcor{d}^{(\sigma)}\coloneqq \dist{\sigma\cdot h}{\subcor{t}\cdot\partial\cnorm{v^{\star}}}\;,\]
and let $\subsig{w}^{(\sigma)}\in\partial\snorm{x^{\star}}$ and $\subcor{w}^{(\sigma)}\in\partial\cnorm{v^{\star}}$ be such that
\[\norm{\sigma\cdot g - \subsig{t}\cdot \subsig{w}^{(\sigma)}}_2 = \subsig{d}^{(\sigma)}\]
and
\[\norm{\sigma\cdot h - \subcor{t}\cdot \subcor{w}^{(\sigma)}}_2 = \subcor{d}^{(\sigma)}\;.\]
We know that for each choice of signs $\sigma_1,\sigma_2\in\{\pm1\}$,
\begin{align*}
0&\geq \left(\snorm{x^{\star}+a}+\lambda\cnorm{v^{\star}+b}\right)\\
&\quad-\left(\snorm{x^{\star}}+\lambda\cnorm{v^{\star}}\right)\\
&\geq \inner{\subsig{w}^{(\sigma_1)}}{a}+\lambda\inner{\subcor{w}^{(\sigma_2)}}{b}\\
& = \frac{\inner{\sigma_1\cdot g}{a} - \inner{\sigma_1\cdot g - \subsig{t}\cdot \subsig{w}^{(\sigma_1)}}{a}}{\subsig{t}} \\&\quad+ \lambda\frac{\inner{\sigma_2 \cdot h}{b} - \inner{\sigma_2\cdot h - \subcor{t}\cdot \subcor{w}^{(\sigma_2)}}{b}}{\subcor{t}}\\
&\geq\frac{\inner{\sigma_1\cdot g}{a}-\subsig{d}^{(\sigma_1)}\norm{a}_2+\inner{\sigma_2 \cdot h}{b}-\subcor{d}^{(\sigma_2)}\norm{b}_2}{\subsig{t}}\;,\end{align*}
and so
\[\subsig{d}^{(\sigma_1)}\norm{a}_2+\subcor{d}^{(\sigma_2)}\norm{b}_2 \geq \inner{\sigma_1 \cdot g}{a}+\inner{\sigma_2 \cdot h}{b}\;.\]
Maximizing the right-hand side over the signs,
\begin{multline}\label{eqn:positive_parts_gaus}\subsig{d}\norm{a}_2+\subcor{d}\norm{b}_2\\\geq |\inner{g}{a}|+|\inner{h}{b}|\geq (\inner{-g}{a})_++(\inner{-h}{b})_+\;,\end{multline}
where
\[\subsig{d}=\max\{\subsig{d}^{(+1)},\subsig{d}^{(-1)}\}\]
and
\[\subcor{d}=\max\{\subcor{d}^{(+1)},\subcor{d}^{(-1)}\}\;.\]

Next, since $\norm{b}_2\leq 1$ for all $(a,b)\in \jtcone\cap\sphere{p+n}$, we have the following (here $(a,b)$ is assumed to be in $\jtcone\cap\sphere{p+n}$)
\begin{align*}&\min_{(a,b)} \norm{h\cdot \norm{a}_2+\sqrt{n}\cdot b}_2+\inner{g}{a}\\
&\geq
 \min_{(a,b)} \norm{h\cdot \norm{a}_2+\norm{h}_2\cdot b}_2+\inner{g}{a}  -\left|\sqrt{n}-\norm{h}_2\right|\\
 & = \min_{(a,b)} \bigg\{\sqrt{\norm{h}^2_2\left(\norm{a}^2_2+\norm{b}^2_2\right) + 2\norm{h}_2\norm{a}_2\cdot\inner{h}{b}}\\
 &\quad +\inner{g}{a}-\left|\sqrt{n}-\norm{h}_2\right|\bigg\}\\
 & = \min_{(a,b)}\bigg\{\norm{h}_2 \sqrt{1 + 2\norm{a}_2\cdot\inner{\frac{h}{\norm{h}_2}}{b}}\\
 &\quad +\inner{g}{a}-\left|\sqrt{n}-\norm{h}_2\right|\bigg\}\\
 & \geq \min_{(a,b)}\bigg\{\norm{h}_2 \sqrt{1- 2\norm{a}_2\cdot\left(\inner{\frac{-h}{\norm{h}_2}}{b}\right)_+}\\
 &\quad-\inner{-g}{a}-\left|\sqrt{n}-\norm{h}_2\right|\bigg\}\\
\intertext{Since $\sqrt{1-x}\geq 1-x$ for all $x\in[0,1]$,}
 & \geq \min_{(a,b)}\bigg\{\norm{h}_2 \left(1- 2\norm{a}_2\cdot\left(\inner{\frac{-h}{\norm{h}_2}}{b}\right)_+\right)\\
 &\quad-\inner{-g}{a}-\left|\sqrt{n}-\norm{h}_2\right|\bigg\}\\
 & = \min_{(a,b)}\bigg\{\norm{h}_2 - 2\norm{a}_2\cdot\left(\inner{-h}{b}\right)_+\\
 &\quad-\inner{-g}{a}-\left|\sqrt{n}-\norm{h}_2\right|\bigg\}\\
& \geq \min_{(a,b)}\bigg\{\norm{h}_2 - 2\norm{a}_2\cdot\left(\inner{-h}{b}\right)_+\\
 &\quad-(\inner{-g}{a})_+-\left|\sqrt{n}-\norm{h}_2\right|\bigg\}\\
\intertext{Applying \eqref{eqn:positive_parts_gaus},}
& \geq \min_{(a,b)}\bigg\{\norm{h}_2 - \max\{2\norm{a}_2,1\}\cdot\\
 &\quad\left(\subsig{d}\norm{a}_2+\subcor{d}\norm{b}_2 - (\inner{-g}{a})_+\right)\\
 &\quad-(\inner{-g}{a})_+-\left|\sqrt{n}-\norm{h}_2\right|\bigg\}\\
& \geq \min_{(a,b)}\bigg\{\norm{h}_2 - \max\{2\norm{a}_2,1\}\cdot\\
 &\quad\left(\subsig{d}\norm{a}_2+\subcor{d}\norm{b}_2 \right)-\left|\sqrt{n}-\norm{h}_2\right|\bigg\}\\
& = \min_{(a,b)}\bigg\{\norm{h}_2 - \max\{2\norm{a}^2_2,\norm{a}_2\}\cdot\subsig{d}\\
 &\quad -\max\{2\norm{a}_2\norm{b}_2,\norm{b}_2\}\cdot \subcor{d}-\left|\sqrt{n}-\norm{h}_2\right|\bigg\}\\
\intertext{Since $2\norm{a}_2\cdot \norm{b}_2\leq 1$ (because $\norm{a}^2_2+\norm{b}^2_2=1$),}
& \geq \norm{h}_2 - 2\subsig{d}-\subcor{d}-\left|\sqrt{n}-\norm{h}_2\right|\;.
\end{align*}
Taking expectations, we have
\begin{multline*}\EE{\min_{(a,b)\in\jtcone\cap\sphere{p+n}} \norm{h\cdot \norm{a}_2+\sqrt{n}\cdot b}_2+\inner{g}{a}}\\\geq \mu_n - 2\gdist{\subsig{t}\cdot\partial\snorm{x^{\star}}}-\gdist{\subcor{t}\cdot\partial\cnorm{v^{\star}}}\\-3\sqrt{2\pi}-\frac{1}{\sqrt{2}}\;,\end{multline*}
where we use \lemref{GaussNorm} to bound $\EE{\left|\sqrt{n}-\norm{h}_2\right|}$, and to bound $\EE{\subsig{d}}$ and $\EE{\subcor{d}}$ we use the following lemma (proved in \appref{lemmas}):
\begin{lemma}\label{lem:MaxDistanceGaussian}
For any set $A\subset\R^p$, for $g\in\R^p$ with \iid standard normal entries,
\[\EE{\max\{\dist{g}{A},\dist{-g}{A}\}}\leq \gdist{A}+\sqrt{2\pi}\;.\]
\end{lemma}

\end{proof}

\section{Proofs for lemmas}\label{app:lemmas}

\begin{replemma}{lem:NuclearNormRandom}
Let $\Gamma\in\R^{m_1\times m_2}$ have \iid standard Gaussian entries, with $m_1\geq m_2$. Then
\[\EE{\norm{\Gamma}_*}\geq \frac{4}{27}m_2\sqrt{m_1}\;.\]
\end{replemma}
\begin{proof}
Let $k=\left\lceil\frac{4}{9}m_2\right\rceil\leq m_2$. Let $\Upsilon$ be a $m_1\times k$ submatrix of $\Gamma$; then $\Upsilon$ also has \iid Gaussian entries, and $\norm{\Gamma}_*\geq \norm{\Upsilon}_*$, so we only need to find a lower bound for $\EE{\norm{\Upsilon}_*}$.

Let $\sigma_1\geq\dots\geq\sigma_k\geq 0$ be the singular values of $h$. \citet[Thm.~II.13]{DavidsonSz08} establishes that 
\[\EE{\sigma_k}\geq \sqrt{m_1}-\sqrt{k}\;.\]
This means that $\EE{\sigma_i}\geq \sqrt{m_1}-\sqrt{k}$ for each $i=1,\dots,k$.
In the simple case that $\frac{4}{9}m_2$ is an integer, we then have
\begin{multline*}\EE{\sigma_1+\dots+\sigma_k}
\geq k(\sqrt{m_1}-\sqrt{k})
\\=\frac{4}{9}m_2\left(\sqrt{m_1}-\frac{2}{3}\sqrt{m_2}\right)\\
\\\geq \frac{4}{9}m_2\left(\sqrt{m_1}-\frac{2}{3}\sqrt{m_1}\right)
=\frac{4}{27}m_2\sqrt{m_1}\;.\end{multline*}

In the general case where $\frac{4}{9}m_2$ might not be an integer, we also use the fact that
\[\EE{\sigma_1}\geq \EE{\norm{\Upsilon_{(1)}}_2} = \mu_{m_1}\;,\]
where $\Upsilon_{(1)}$ is the first column of $\Upsilon$. Since $\mu_{m_1} > \sqrt{m_1-1/2}$, which is larger than our previous lower bound $\sqrt{m_1}-\sqrt{k}$ on $\EE{\sigma_1}$ used above, this gives us the slack we need in order to handle the slight discrepancy between $k$ and $\frac{4}{9}m_2$. We omit the details.

\end{proof}

\begin{replemma}{lem:NormRatios}
For any norm $\norm{\cdot}$ on $\R^p$ with dual norm $\norm{\cdot}^*$,
\[\max_{x\in\R^p}\frac{\norm{x}}{\norm{x}_2} = \max_{x\in\R^p}\frac{\norm{x}_2}{\norm{x}^*}\;.\]
\end{replemma}
\begin{proof}
First, for any $x$,
\[\norm{x}^2_2 = \inner{x}{x}\leq \norm{x}\cdot\norm{x}^* \ \Rightarrow \ \frac{\norm{x}}{\norm{x}_2}\geq \frac{\norm{x}_2}{\norm{x}^*}\;.\]
This proves that the left-hand side is greater than or equal to the right-hand side, in the claim.
Now we prove the reverse inequality. Choose any $x\in\R^p$. Then $\norm{x}=\max_{\norm{w}^*=1}\inner{x}{w}$. Choose some $w$ with $\norm{w}^*=1$ so that this maximum is attained. Then $\inner{x}{w}\leq\norm{x}_2\cdot\norm{w}_2$, and so
\[\frac{\norm{x}}{\norm{x}_2} = \frac{\inner{x}{w}}{\norm{x}_2}\leq \norm{w}_2 \leq \max_{x'\in\R^p}\frac{\norm{x'}_2}{\norm{x'}^*}\;.\]
Since this is true for any $x$, this proves the desired bound.
\end{proof}

\begin{replemma}{lem:t-hat-Lipschitz}
Suppose that, for $x\neq 0$, $\partial\norm{x}$ satisfies \eqref{eqn:subdiff-decomp}. Let $t_g$ be defined as in \eqref{eqn:best-t-for-g}. 
Then $g\mapsto  t _g$ is a $\frac{1}{\norm{w_0}_2}$-Lipschitz function of $g$.
\end{replemma}
\begin{proof}
Consider any $g$ and $g'$. Find $w,w'\in\partial\norm{x}$ such that
\[\dist{g}{ t _g\cdot\partial\norm{x}} = \norm{g- t _g\cdot w}_2\]
and
\[\dist{g'}{ t _{g'}\cdot\partial\norm{x}} = \norm{g'- t _{g'}\cdot w'}_2\;.\]
Recalling that
\[\ncone=\cup_{ t \geq 0}( t \cdot\partial\norm{x})\;,\]
we see that $ t _g\cdot w$ and $ t _{g'}\cdot w'$ are the projections of $g$ and of $g'$, respectively, onto the cone $\ncone$. Since $\ncone$ is convex and projection onto a convex set is nonexpansive, it follows that
\[\norm{ t _g\cdot w- t _{g'}\cdot w'}_2\leq\norm{g-g'}_2\;.\]
Now we use the assumption \eqref{eqn:subdiff-decomp}: we have
\begin{align*}
&\norm{ t _g\cdot w- t _{g'}\cdot w'}_2 \\
&= \norm{( t _g- t _{g'})\cdot w_0 + \left[ t _g\cdot(w-w_0)- t _{g'}\cdot(w'-w_0)\right]}_2\\
 &\geq \left| t _g- t _{g'}\right|\cdot\norm{w_0}_2\;,\end{align*}
because $\inner{w-w_0}{w_0}=\inner{w'-w_0}{w_0}=0$.
\end{proof}

\begin{replemma}{lem:Lipschitz-E-vs-P}
Let $\phi:\R^p\rightarrow \R$ be a $1$-Lipschitz function. For any $a\in\R$ and $p_0>0$, 
\[\PP{\phi(g)< a}\geq p_0 \quad \Rightarrow \quad \EE{\phi(g)}\leq a + \sqrt{2\log(\nicefrac{1}{p_0})}\;.\]
\end{replemma}
\begin{proof}
Choose any $\eps$ satisfying $0<\eps<p_0$. By \citet[(2.8)]{Ledoux}, we know that
\[\PP{\phi(g)<\EE{\phi(g)} - \sqrt{2\log\left(\frac{1}{p_0-\eps}\right)}}\leq  p_0-\eps\;,\]
which means that we must have $\EE{\phi(g)}- \sqrt{2\log(\frac{1}{p_0-\eps})}<a$ to avoid contradiction. Since this is true for arbitrary $\eps\in(0,p_0)$, the claim follows.
\end{proof}

\begin{replemma}{lem:GaussianInequality}
Take any 
$\Omega\subset\twoball{p}\times \R^n$.
Let $\Phi\in\R^{n\times p}$ have \iid $N(0,\frac{1}{n})$ entries, and let $g\in\R^p$ and $h\in\R^n$ have \iid $N(0,1)$ entries.
Then
\begin{multline*}\sqrt{n}\cdot \EE{\min_{(a,b)\in\Omega}\norm{\Phi a + b}_2} \geq\\ \EE{\left(\min_{(a,b)\in\Omega}\norm{h\cdot\norm{a}_2+\sqrt{n}\cdot b}_2 + \inner{g}{a}\right)_+} -\frac{1}{\sqrt{2\pi}}\;.\end{multline*}
\end{replemma}

\begin{proof}
For $(a,b)\in\Omega$ and $w\in\twoball{n}$, let
\[X_{(a,b),w}=\sqrt{n}\cdot \inner{\Phi a}{w} + \nu\norm{a}_2\]
and
\[Y_{(a,b),w}=\inner{h}{w}\norm{a}_2+\inner{g}{a}\;,\]
where $\nu\sim N(0,1)$ is independent from the other random variables.
Then $\{X_{(a,b),w}\}$ and $\{Y_{(a,b),w}\}$ are both centered Gaussian processes, with
\[\EE{X_{(a,b),w}^2} = \EE{Y_{(a,b),w}^2}=\norm{a}^2_2(\norm{w}^2_2+1)\]
and
\begin{align*}
&\EE{X_{(a,b),w}X_{(a',b'),w'}} - \EE{Y_{(a,b),w}Y_{(a',b'),w'}}\\
 &= \left[\inner{a}{a'}\inner{w}{w'}+\norm{a}_2\norm{a'}_2\right]\\&\quad-\left[\norm{a}_2\norm{a'}_2\inner{w}{w'} + \inner{a}{a'}\right]\\
&=\left(\norm{a}_2\norm{a'}_2-\inner{a}{a'}\right)\cdot\left(1-\inner{w}{w'}\right)
\geq 0\;,
\end{align*}
with equality if $(a,b)=(a',b')$. Therefore, applying Theorem 1.1 of \citet{Gordon} to these Gaussian processes, we know that for any scalars $\{c_{(a,b),w}\}$,
\begin{multline*}\PP{\cap_{(a,b)\in\Omega}\cup_{w\in\twoball{n}} \left[X_{(a,b),w}\geq c_{(a,b),w}\right]}\\\geq \PP{\cap_{(a,b)\in\Omega}\cup_{w\in\twoball{n}} \left[Y_{(a,b),w}\geq c_{(a,b),w}\right]}\;.\end{multline*}
Now fix any $C\in\R$ and let $c_{(a,b),w}=C-\sqrt{n}\cdot \inner{b}{w}$. Then we can simplify these events:
\begin{multline*}\cap_{(a,b)\in\Omega}\cup_{w\in\twoball{n}} \left[X_{(a,b),w}\geq c_{(a,b),w}\right] \\=\left\{\min_{(a,b)\in\Omega}\max_{w\in\twoball{n}}X_{(a,b),w}+\sqrt{n}\cdot \inner{b}{w}\geq C\right\}\end{multline*}
and same for the $Y$ process. This proves that
\begin{multline*}\PP{\min_{(a,b)\in\Omega}\max_{w\in\twoball{n}}X_{(a,b),w}+\sqrt{n}\cdot \inner{b}{w}\geq C}\\\geq \PP{\min_{(a,b)\in\Omega}\max_{w\in\twoball{n}}Y_{(a,b),w}+\sqrt{n}\cdot \inner{b}{w}\geq C}\;,\end{multline*}
or in other words, by maximizing over $w$ on each side,
\begin{multline*}\PP{\min_{(a,b)\in\Omega}\sqrt{n}\cdot\norm{\Phi a + b}_2 + \nu\cdot\norm{a}_2\geq C}\\\geq \PP{\min_{(a,b)\in\Omega}\norm{h\cdot\norm{a}_2+\sqrt{n}\cdot b}_2 + \inner{g}{a}\geq C}\;.\end{multline*}
Since this is true for all $C\in\R$, we can integrate over $C\in(0,\infty)$ to obtain
\begin{multline*} \EE{\left(\min_{(a,b)\in\Omega}\sqrt{n}\cdot\norm{\Phi a + b}_2 + \nu\cdot\norm{a}_2\right)_+}\\\geq \EE{\left(\min_{(a,b)\in\Omega}\norm{h\cdot\norm{a}_2+\sqrt{n}\cdot b}_2 + \inner{g}{a}\right)_+}\;.\end{multline*}

Finally, since $\norm{a}_2\in[0,1]$ for all $(a,b)\in\Omega$, we see that
\begin{multline*}\left(\min_{(a,b)\in\Omega}\sqrt{n}\cdot\norm{\Phi a + b}_2 + \nu\cdot\norm{a}_2\right)_+\\ \leq \sqrt{n}\cdot\min_{(a,b)\in\Omega}\norm{\Phi a + b}_2+ (\nu)_+\;,\end{multline*}
and therefore, since $\EE{(\nu)_+}=\frac{1}{\sqrt{2\pi}}$, 
\begin{multline*}\sqrt{n}\cdot\EE{\min_{(a,b)\in\Omega}\norm{\Phi a + b}_2}\geq \\\EE{\min_{(a,b)\in\Omega} \left(\sqrt{n}\cdot\norm{\Phi a + b}_2 + \nu\cdot\norm{a}_2\right)}-\frac{1}{\sqrt{2\pi}}\;.\end{multline*}

\end{proof}

\begin{replemma}{lem:abc_Identity}
For any $a,b,C\geq 0$ such that $a^2+b^2\leq C^2$,
\[\inf_{u\in[0,1]} \sqrt{C^2-2Cu\sqrt{1-u^2}\cdot a}-u\cdot b\geq C - \sqrt{a^2+b^2}\;.\]
\end{replemma}
\begin{proof}
By rescaling, we can assume without loss of generality that $C=1$ and $a^2+b^2\leq 1$. First, suppose $a^2+b^2=1$. Then
\begin{align*}
&(u\cdot a - \sqrt{1-u^2})^2 \geq 0\\
\Rightarrow\quad&u^2 \cdot a^2 -2 u\sqrt{1-u^2}\cdot a + (1-u^2)\geq 0\\
\Rightarrow \quad&1 -2 u\sqrt{1-u^2}\cdot a \geq u^2\cdot (1-a^2) = u^2\cdot b^2\\
\Rightarrow \quad&\sqrt{1 -2 u\sqrt{1-u^2}\cdot a} - u\cdot b\geq 0\;.
\end{align*}
This proves the claim whenever $a^2+b^2=1$. Now suppose that $a^2+b^2=c^2$ for some $c\in[0,1]$, and let $a'=\frac{a}{c}$ and $b'=\frac{b}{c}$. Then
\begin{align*}
&\sqrt{1 -2 u\sqrt{1-u^2}\cdot a} - u\cdot b\\
&=\sqrt{1 -2 u\sqrt{1-u^2}\cdot a'\cdot c} - u\cdot b'\cdot c\\
&=\left[\sqrt{1 -2 u\sqrt{1-u^2}\cdot a'\cdot c} - c \sqrt{1-2u\sqrt{1-u^2}\cdot a'}\right] \\
&\quad + c\cdot\left[\sqrt{1-2u\sqrt{1-u^2}\cdot a'} - u\cdot b'\right]\\
\intertext{From the work above, since $a'{}^2+b'{}^2=1$,}
&\geq\sqrt{1 -2 u\sqrt{1-u^2}\cdot a'\cdot c} - c\cdot \sqrt{1-2u\sqrt{1-u^2}\cdot a'}\\
\intertext{Writing $d=2 u\sqrt{1-u^2}\cdot a'\in [0,1]$,}
&=\sqrt{1 - d\cdot c} - c\cdot \sqrt{1-d}\\
&\geq 1-c\;,
\end{align*}
where the last step is true for all $c,d\in[0,1]$ as follows:
\begin{align*}
&(1-\sqrt{1-d})^2 = 2-2\sqrt{1-d}-d\\
\intertext{Since $c^2\leq c$,}
\Rightarrow \quad & c^2(1-\sqrt{1-d})^2\leq c(2(1-\sqrt{1-d})-d)\\
\intertext{By rearranging terms and adding $1$ to both sides,}
\Rightarrow \quad & 1+c^2(1-\sqrt{1-d})^2 -2c (1-\sqrt{1-d})\leq 1-dc\\
\intertext{By taking the square root,}
\Rightarrow \quad & 1 - c(1-\sqrt{1-d})\leq \sqrt{1-dc}\\
\Rightarrow \quad & 1 - c\leq \sqrt{1-dc}-c\sqrt{1-d}\;.
\end{align*}
\end{proof}

\begin{replemma}{lem:GaussNorm}
For $h\sim N(0,I_n)$,
$\EE{\left|\norm{h}_2-\sqrt{n}\right|}\leq \frac{1}{\sqrt{2}}$.
\end{replemma}
\begin{proof}
We use the fact that $\EE{\norm{h}_2}=\mu_n$. We have
\begin{align*}
&\EE{\left|\norm{h}_2-\sqrt{n}\right|}^2\\
&\leq \EE{(\norm{h}_2-\sqrt{n})^2}\\
&=\EE{(\norm{h}_2-\mu_n)^2}+(\mu_n-\sqrt{n})^2\\
&=(n-2\mu_n^2+\mu_n^2)+(\mu_n^2-2\mu_n\sqrt{n}+n)\\
&=2\sqrt{n}(\sqrt{n}-\mu_n)\\
&\leq \frac{1}{2}\;,\end{align*}
where the last inequality holds for all $n\geq 1$.

\end{proof}

\begin{replemma}{lem:MaxDistanceGaussian}
For any set $A\subset\R^p$, for $g\in\R^p$ with \iid standard normal entries,
\[\EE{\max\{\dist{g}{A},\dist{-g}{A}\}}\leq \gdist{A}+\sqrt{2\pi}\;.\]
\end{replemma}
\begin{proof}
Let $\alpha=\EE{\dist{g}{A}}=\EE{\dist{-g}{A}}\leq \sqrt{\EE{\dist{g}{A}^2}}=\gdist{A}$. Since $g\mapsto\dist{g}{A}$ is $1$-Lipschitz,
\begin{multline*}\EE{(\dist{g}{A}-\alpha)_+}=\\\int_{0}^{\infty}\PP{\dist{g}{A}>\alpha+c}\ \diff{c}\\= \int_{0}^{\infty}e^{-c^2/2}\ \diff{c}\leq\sqrt{\nicefrac{\pi}{2}}\;.\end{multline*}
Then
\begin{align*}
&\EE{\max\{\dist{g}{A},\dist{-g}{A}\}}\\
&=\alpha+\EE{\max\{\left(\dist{g}{A}-\alpha\right),\left(\dist{-g}{A}-\alpha\right)\}}\\
&\leq \alpha \!+\! \EE{\left(\dist{g}{A}-\alpha\right)_+}\! + \!\EE{\left(\dist{-g}{A}-\alpha\right)_+}\\
&\leq \alpha + \sqrt{2\pi}\leq \gdist{A}+\sqrt{2\pi}\;.
\end{align*}
\end{proof}

\section*{Acknowledgments}
The authors thank Emmanuel Cand{\`es} for helpful suggestions on the presentation of this work.
R.F.\ was supported by NSF grant DMS-1203762.

\bibliography{notes.bib}

\end{document}